\documentclass[11pt,a4paper,onepage]{amsart}
\usepackage{eurosym}
\usepackage{amsfonts}
\usepackage[linktocpage=true,colorlinks=true,citecolor=violet,linkcolor=blue,urlcolor=blue]{hyperref}
\usepackage{enumerate}
\usepackage{ytableau}

\numberwithin{equation}{section}
\newcommand{\dd}{\mathrm{d}}

\newtheorem{theorem}{Theorem}[section]
\newtheorem{proposition}[theorem]{Proposition}

\theoremstyle{definition}
\newtheorem{corollary}[theorem]{Corollary}

\theoremstyle{remark}
\newtheorem{remark}[theorem]{Remark}

\begin{document}
\title{Schur expansion of random-matrix reproducing kernels}
\author[]{Leonardo Santilli}
\address[LS]{Departamento de Matem\'{a}tica, Grupo de F\'{i}sica Matem\'{a}tica, Faculdade de Ci\^{e}ncias, Universidade de Lisboa, Campo Grande, Edif\'{i}cio C6, 1749-016 Lisboa, Portugal.}
\email{lsantilli@fc.ul.pt}
\author[]{Miguel Tierz}
\address[MT]{Departamento de Matem\'{a}tica, ISCTE - Instituto Universit\'{a}rio de Lisboa, Avenida das For\c{c}as Armadas, 1649-026 Lisboa, Portugal.}
\email{mtpaz@iscte-iul.pt}
\address[MT]{Departamento de Matem\'{a}tica, Grupo de F\'{i}sica Matem\'{a}tica, Faculdade de Ci\^{e}ncias, Universidade de Lisboa, Campo Grande, Edif\'{i}cio C6, 1749-016 Lisboa, Portugal.}
\email{tierz@fc.ul.pt}

\begin{abstract}
We give expansions of reproducing kernels of the Christoffel--Darboux type in terms of Schur polynomials. 
For this, we use evaluations of averages of characteristic polynomials and Schur polynomials in random matrix ensembles. We explicitly compute new Schur averages, such as the Schur average in a $q$-Laguerre ensemble, and the ensuing expansions of random matrix kernels. In addition to classical and $q$-deformed cases on the real line, we use extensions of Dotsenko--Fateev integrals to obtain expressions for kernels on the complex plane. Moreover, a known interplay between Wronskians of Laguerre polynomials, Painlev\'{e} tau functions and conformal block expansions is discussed in relationship to the Schur expansion obtained. 
\end{abstract}

\maketitle
\tableofcontents

\section{Introduction}

The study of reproducing kernels \cite{Mercer,Arons} is of significance
across a remarkably transversal spectrum of areas and applications.
Reproducing kernels appear in a crucial way in many different areas including
random matrix theory \cite{Forrester}, machine learning \cite{Kernel-ML},
quantization and the study of coherent states \cite{Coherent}, Shannon
sampling theorems, non-parametric density estimation in statistics and probability theory \cite{Parzen,book1}.\par
Typical in random matrix theory are the so-called  Christoffel--Darboux kernels. The most immediate meaning of these kernels is the description of the correlations between two eigenvalues of a random matrix \cite{Forrester}. The reproducing property allows to express higher-order correlation functions in terms of determinants of the two-point kernel, leading to the consideration of a multi-point kernel. In this work we will study the latter, and the method employed lends results for this more general kernel with no additional difficulty. In the simpler setting of the two-point kernel, we will show that the Schur expansion can also be written in terms of Chebyshev polynomials of the second kind.

Christoffel--Darboux kernels are a well-known type of delta sequences and, as
such, had found since quite some time important applications in statistics 
\cite{WB}. These kernels are also used in modern machine learning and
data-analysis contexts \cite{CDML}, including for example data-driven
studies of approximations of the spectrum of the Koopman operator \cite{Koopman}. For these reasons, we will comment at the end how the very recently studied antisymmetric kernels \cite{Klus} are related to the
discussion here (and a complementary discussion in \cite{ST}) and to certain specific models discussed in this work, mainly the Stieltjes--Wigert model.\par
\medskip
The paper is organized as follows. We start with the basic definitions and introduce notation and conventions. Section \ref{sec:GeneralSchurExp} presents the main idea of the work, summarized in Theorem \ref{thm:main}. Then, in Section \ref{sec:ClassicCDK} we apply the Schur expansion to kernels associated to classical matrix ensembles. For the Laguerre ensemble, a relation with Painlev\'{e} V is revisited in Section \ref{sec:CBExp}, where we also compare our expansion with the conformal block expansion of the kernel. 
In Section \ref{sec:qdefCDK} we focus on kernels associated to $q$-ensembles and in Section \ref{sec:CKernels} on kernels on $\mathbb{C}$. Along the way, we compute the average of a Schur polynomial in the $q$-Laguerre ensemble, whose evaluation constitutes a novel result.\par
 
We conclude with an outlook on possible avenues for further research in Section \ref{sec:outlook}. The text is complemented with three appendices: Appendix \ref{app:ComputeSchur} contains the technical details of the computation of Schur averages, whereas Appendices \ref{app:Duduchava} and \ref{app:HKCheby} exemplify how the Schur expansion method transcends the topic of this work and easily finds application in a broad variety of related problems.

\subsection{Definitions and notation}

\subsubsection{Christoffel--Darboux~kernels}

For any given weight function $w(z)$ we will denote $P_{n}$ the corresponding monic orthogonal polynomial of degree $n$ and adopt the usual notation $h_n \equiv \left\Vert P_n \right\Vert^2 = \langle P_{n}, P_{n} \rangle_w $.

We denote by $N$ the rank of the kernel and by $n$ the number of pairs of variables $(x_1, \dots, x_n; y_1 , \dots, y_n)$. Let $K_N ^{(n)}
(x_1, \dots, x_n; y_1 , \dots, y_n)$ be the $2n$-point kernel, and $K_N
(x;y) \equiv K_N ^{(1)} (x; y)$. We adopt the definition 
\begin{equation*}  \label{eq:defKN}
K_N (x;y) = \sum_{j=0} ^{N-1} \frac{P_j (x) P_j (y)}{ h_j} ,
\end{equation*}
as for instance in \cite[Eq.(3.1.9)]{Szego} and \cite[Eq.(2.2)]{Borodin:1998hxr}, and
differing by a factor $\sqrt{ w(x) w(y)}$ from \cite[Eq.(5.6)]{Forrester}.
The Christoffel--Darboux~formula states that 
\begin{equation*}
K_N (x,y) = \frac{1}{h_{N-1}} \frac{ P_{N} (x) P_{N-1} (y) - P_{N-1} (x) P_N
(y)}{ x-y } .
\end{equation*}
The multi-point kernel satisfies 
\begin{equation}  \label{eq:defmultikernel}
K_N ^{(n)} (x_1, \dots, x_n; y_1 , \dots, y_n) = \frac{1}{\Delta_n (x)
\Delta_n (y)} \det_{1 \le i,j \le n} K_N (x_i;y_j) ,
\end{equation}
where 
\begin{equation*}
   \Delta_n (x) \equiv \prod_{1 \le i < j \le n}  (x_i - x_j)
\end{equation*}
is the Vandermonde determinant in $n$ variables (and analogously for $\Delta_n (y) $).\par
We introduce the notation 
\begin{equation*}  \label{eq:defXDualVar}
x^{\vee} \equiv \left( -\frac{1}{x_1}, \dots, - \frac{1}{x_n} \right) ,
\qquad y^{\vee} \equiv \left( -\frac{1}{y_1}, \dots, - \frac{1}{y_n} \right)
,
\end{equation*}
and often combine these parameters into the $2n$-dimensional vector 
\begin{equation}  \label{eq:deft}
t \equiv \left( -\frac{1}{x_1}, \dots, - \frac{1}{x_n} , - \frac{1}{y_1} ,
\dots, - \frac{1}{y_n} \right) .
\end{equation}
Besides, to reduce clutter we define 
\begin{equation}  \label{eq:defhatK}
\widehat{K}_N ^{(n)} (x_1 , \dots, x_n ;y_1 , \dots , y_n) = \frac{
\prod_{j=N-n} ^{N-1} h_j }{\prod_{i=1} ^{n} (x_i y_i)^{N-n} } ~ K_N ^{(n)}
(x_1 ,\dots, x_n;y_1, \dots , y_n) .
\end{equation}\par
The kernel \eqref{eq:defmultikernel} admits an integral representation (see Section \ref{sec:GeneralSchurExp}) in which the number of integration variables is $N-n$. We will denote 
\begin{equation*}
M \equiv N-n
\end{equation*}
the number of variables in a generic ensemble. Moreover, for any given ensemble of $M$ variables with weight function $w(z)$, $\mathcal{Z}_{M}$ denotes the partition function, that is, the $M$-fold integral 
\begin{equation*}
    \mathcal{Z}_{M} \equiv \int \Delta_M (z)^2 \prod_{j=1} ^{M} w (z_j) \dd z_j , 
\end{equation*}
and we will use the shorthand notation 
\begin{equation}  \label{eq:defnormN}
\mathcal{N}_{M} \equiv \frac{1}{M! \mathcal{Z}_{M} }
\end{equation}
for the ubiquitous normalization factor.

\subsubsection{Partitions}

Let $\mathbb{Y}$ be the set of all partitions, 
\begin{equation*}
\mathbb{Y} \equiv \left\{ \lambda = (\lambda_1 , \lambda_2 , \dots) \ \vert
\ \lambda_{j} \ge \lambda_{j+1} \ge 0, \ \forall j \ge 1 \right\}  .
\end{equation*}
The length of a partition $\lambda$ is $\ell (\lambda) = \max \left\{ j \ :
\ \lambda_j >0 \right\}$, and $\vert \lambda \vert = \sum_{j \ge 1}
\lambda_j $ is its size. The transposed partition is denoted $\lambda^{\prime}$.
Besides, for fixed $L,M \in \mathbb{N}$ we define 
\begin{equation}  \label{eq:defBoundPartitions}
\mathbb{Y}_{L,M} \equiv \left\{ \lambda \in \mathbb{Y} \ \vert \ \ell
(\lambda) \le L \text{ and } \lambda_1 \le M \right\} ,
\end{equation}
the set of partitions that are contained in a rectangle of $L$ rows and $M$
columns (not to be confused with the partitions of $L\times M$).

$s_{\lambda} (z_1, \dots, z_M)$ is the Schur polynomial labelled by the
partition $\lambda$ \cite{MacDonald}. For the argument of Schur polynomials, we will often
use shorthand notations $s_{\lambda} (z) \equiv s_{\lambda} (z_1, \dots,
z_M) $ and $s_{\lambda} (1^{M}) \equiv s_{\lambda} (\underbrace{1, \dots, 1}_{\text{$M$ times}})$.

\section{General aspects of the Schur expansion}
\label{sec:GeneralSchurExp}

In this section we set up the method to obtain the Schur expansion of the
kernel $K_N (x;y)$ and its multi-point generalization $K_N ^{(n)} (x_1,
\dots, x_n; y_1 , \dots , y_n)$. Our proof directly gives $\widehat{K}_N ^{(n)}$, as defined in \eqref{eq:defhatK}, from which the kernel is immediately obtained.\par
The starting point is the integral representation \cite[Ch.5]{Forrester}
\begin{align}
\widehat{K}_{N}^{(n)}(x_{1},\dots ,x_{n};y_{1},\dots ,y_{n})& =\frac{%
\mathcal{N}_{M}}{\prod_{i=1}^{n}(x_{i}y_{i})^{M}}\int \Delta
_{M}(z)^{2}\prod_{j=1}^{M}\left[ \prod_{i=1}^{n}\left( x_{i}-z_{j}\right)
\left( y_{i}-z_{j}\right) \right] w(z_{j})\mathrm{d}z_{j}  \notag \\
& =\mathcal{N}_{M}\int \Delta _{M}(z)^{2}\prod_{j=1}^{M}\left[
\prod_{i=1}^{2n}\left( 1+t_{i}z_{j}\right) \right] w(z_{j})\mathrm{d}z_{j},
\label{eq:KNint}
\end{align}%
where in the second line we have used the definition \eqref{eq:deft} of $t$. The overall coefficient $\mathcal{N}_{M}$ has been defined in \eqref{eq:defnormN} and $M=N-n$.\par
The next step is to use the dual Cauchy identity \cite{MacDonald,MacD2}
\begin{equation}
\prod_{i=1}^{2n}\prod_{j=1}^{M}\left( 1+t_{i}z_{j}\right) =\sum_{\lambda \in 
\mathbb{Y}_{2n,M}}s_{\lambda }(t)s_{\lambda ^{\prime }}(z).
\label{eq:dualCauchy}
\end{equation}%
We have used $\ell (\lambda ^{\prime })=\lambda _{1}$, and $\mathbb{Y}%
_{2n,M} $ is the finite set defined in \eqref{eq:defBoundPartitions}.%
\footnote{As argued in \cite[Sec.5]{Santilli:2020snh}, the dual Cauchy formula \eqref{eq:dualCauchy} is an exact polynomial identity, as opposed to the Cauchy identity, that is only algebraic and should be understood in a perturbative sense.} Plugging \eqref{eq:dualCauchy} in \eqref{eq:KNint}
yields the following central result.

\begin{theorem}
\label{thm:main} In the above notation, 
\begin{equation}  \label{eq:kernelexp}
\widehat{K}_N ^{(n)} (x_1, \dots, x_n; y_1, \dots, y_n) =\sum_{ \lambda \in 
\mathbb{Y}_{2n,N-n} } s_{\lambda} (t) \left\langle s_{\lambda^{\prime}}
\right\rangle_{w} .
\end{equation}
\end{theorem}

Therefore, whenever the average of a Schur polynomial in the ensemble
characterized by the weight function $w(z)$ is known, we get an expansion of 
$\widehat{K}_N ^{(n)}$, and thus of $K_N ^{(n)}$, in the Schur basis with
explicitly known coefficients. The symmetry enhancement 
\begin{equation}  \label{eq:symen}
S_{n} \times S_{n} \times \mathbb{Z}_2 \hookrightarrow S_{2n}
\end{equation}
for the parameters $x^{\vee}, y^{\vee}$ is manifest in the expansion.\par
\medskip
According to the Schur-reproducing property \cite{Mironov:2017och,Mironov:2018ekq}, each summand in the Schur expansion will take the schematic form $c_{N-n} (\lambda^{\prime}) s_{\lambda^{\prime}} (1^{N-n}) s_{\lambda} (t) $, for some partition-dependent coefficient $c_{N-n} (\lambda^{\prime})$. This property of the average is robust under various layers of deformations \cite{Mironov:2018ekq,Morozov:2018eiq,Mironov:2020jyo} and is inherited by the corresponding kernel. Concretely, 
\begin{itemize}
    \item In a classical ensemble, $c_{N-n} (\lambda^{\prime})$ will be a rational function of the rows $\lambda_j ^{\prime}$, typically expressed in terms of $\Gamma$-functions.
    \item In a $q$-ensemble, $c_{N-n} (\lambda^{\prime})$ will be a rational function of $q$ with exponential dependence on the rows $\lambda_j ^{\prime}$, possibly times a rational function of the $\lambda_j ^{\prime}$.
\end{itemize}
This general observation will be manifest in the explicit results of the following sections.\par

\begin{remark}
\label{rmk:betaensemble} The dual Cauchy identity \eqref{eq:dualCauchy} yields the expansion in Schur polynomials of the correlation function of characteristic polynomials in any beta-ensemble. However, denoting the beta-parameter by $2 \gamma$, it is more convenient in the $\gamma \ne 1$ setup to use the alternative dual Cauchy identity \cite{MacD2}
\begin{equation}
\label{eq:DCJack}
   \prod_{i=1}^{2n}\prod_{j=1}^{M}\left( 1+t_{i}z_{j}\right) =\sum_{\lambda \in 
\mathbb{Y}_{2n,M}} P^{(\gamma)}_{\lambda }(t) P^{(1/\gamma)}_{\lambda ^{\prime }}(z) ,
\end{equation}
where $P^{(\gamma)} _{\lambda}$ are the Jack polynomials \cite{MacDonald,Forrester}. This aspect is discussed explicitly in Subsection \ref{sec:JUEavg} below.
\end{remark}

It follows directly from Theorem \ref{thm:main} that the 2-point kernel $\widehat{K}_N$ admits an expansion involving Chebyshev polynomials of the
second kind, through the relation 
\begin{subequations}
\begin{align}
\widehat{K}_N (x;y) & = \sum_{\lambda_1 =0}^{N-2} \sum_{\lambda_2 =0}
^{\lambda_1} s_{\lambda} (t_1,t_2) \left\langle s_{\lambda^{\prime}}
\right\rangle_{w}  \notag \\
& = \sum_{\lambda_1 =0}^{N-2} \sum_{\lambda_2 =0} ^{\lambda_1} \left\langle
s_{\lambda^{\prime}} \right\rangle_{w} (t_1 t_2)^{\lambda_2} \sum_{j=0}
^{\lambda_1 - \lambda_2} t_1^{\lambda_1 - \lambda_2 -j} t_2 ^{j}
\label{eq:Chebyl1} \\
& = \sum_{\lambda_1 =0}^{N-2} \sum_{\lambda_2 =0} ^{\lambda_1} \left\langle
s_{\lambda^{\prime}} \right\rangle_{w} (\sqrt{t_1 t_2})^{\vert \lambda \vert
} U_{\lambda_1 - \lambda_2} \left( \frac{t_1 + t_2}{2 \sqrt{t_1 t_2}} \right)
\label{eq:Chebyl2} \\
& = \sum_{\lambda_1 =0}^{N-2} \sum_{\lambda_2 =0} ^{\lambda_1} \left\langle
s_{\lambda^{\prime}} \right\rangle_{w} (xy)^{- \frac{\vert \lambda \vert}{2}
} U_{\lambda_1 - \lambda_2} \left( - \frac{x + y}{2 \sqrt{xy}} \right) .
\label{eq:Chebyl3}
\end{align}
\end{subequations}
Equality \eqref{eq:Chebyl1} follows from straightforward computation of the ratio of $2 \times 2$ determinants 
\begin{equation*}
    s_{\lambda} (t_1,t_2) = \frac{1}{t_1-t_2} \det \left(  \begin{matrix} t_1 ^{\lambda_1 +1} & t_1 ^{\lambda_2} \\  t_2 ^{\lambda_1 +1} & t_2 ^{\lambda_2} \\  \end{matrix} \right) = (t_1 t_2)^{\lambda_2} \frac{ t_1 ^{\lambda_1 - \lambda_2 +1} -  t_2 ^{\lambda_1 - \lambda_2 +1} }{t_1 - t_2 } 
\end{equation*}
and \eqref{eq:Chebyl2} by identification with the Chebyshev polynomial of the second kind (cf. \cite{KadellSC}). Eventually, in \eqref{eq:Chebyl3} we have written the result in the variables $(x;y)$, related to $(t_1,t_2)$ via \eqref{eq:deft}.\par
Note that the $(x;y)$-independent coefficient in the expansion is the same for the Schur
and Chebyshev expansion, since the manipulations only involve the symmetric
polynomials.

\begin{remark}
We emphasize that the 2-point kernel admits an expansion in Chebyshev polynomials, whose orthogonality relations are on the interval $[-1,1]$,
regardless of the domain of the variables $x,y$. Indeed, we did not use the
orthogonality property of these polynomials in the derivation. Notice
however the important extra factor $(xy)^{-\lvert \lambda \rvert /2}$, that is, the dependence is not entirely captured by the Chebyshev polynomial.
\end{remark}
\par
\medskip 
From \eqref{eq:KNint} we may equivalently expand for $(x_1, \dots,
x_n)$ and $(y_1, \dots, y_n)$ separately, applying \eqref{eq:dualCauchy}
twice.

\begin{theorem}\label{thm:double}
In the above notation, 
\begin{equation}  \label{eq:KernelDoubleexp}
\widehat{K}_N ^{(n)} (x_1, \dots, x_n; y_1, \dots, y_n ) = \sum_{ \lambda,
\mu \in \mathbb{Y}_{n,N-n} } s_{\lambda} \left(x^{\vee} \right) s_{\mu}
\left(y^{\vee} \right) \left\langle s_{\lambda^{\prime}} s_{\mu^{\prime}}
\right\rangle_{w} .
\end{equation}
\end{theorem}

This latter expansion was first obtained by Rosengren \cite[Prop.5]{Rosen}.
The equivalence between \eqref{eq:KernelDoubleexp} and Rosengren's formula
stems from $\left\langle s_{\lambda^{\prime}} s_{\mu^{\prime}}
\right\rangle_{w}$ being equal to the determinant of a Hankel minor, and rearranging the terms in the sum.\par
Thus, the result of Theorem \ref{thm:double} is known, but the derivation here is different, only using the dual Cauchy identity. In addition, the identification with Rosengren's previous result is not completely immediate and it involves the application of Andreief's identity \cite{Andreief}.\par

In the 2-point case, Theorem \ref{thm:double} implies the following
classical result \cite{Collar,Kailath} (see also \cite%
{Berg,SimonCD,Borodin:1998hxr}).

\begin{corollary}
\label{cor:GenHankel} Let $\mathcal{H}_N (w) $ the $N \times N$ Hankel
matrix of moments of the measure $w(z) \mathrm{d} z$, 
\begin{equation*}
\left[ \mathcal{H}_N (w) \right]_{jk} = \int_{\mathbb{R}} z^{j+k} w(z) 
\mathrm{d} z , \qquad j,k \in \left\{  0,1, \dots, N-1 \right\},
\end{equation*}
and denote by $\mathcal{H}^{-1} _N (w)$ its inverse. The kernel $K_N
(x;y)$ is the generating function of $\mathcal{H}^{-1} _N (w)$, 
\begin{equation*}
K_N (x;y) = \sum_{j=0}^{N-1} \sum_{k=0} ^{N-1} x^{j} y^{k} \left[ \mathcal{H}^{-1} (w) \right]_{jk} .
\end{equation*}
\end{corollary}

\begin{proof}
Set $n=1$ in Theorem \ref{thm:double}. We recognise in the coefficient $\left\langle s_{(1^j)} s_{(1^k)}
\right\rangle_{w} = \left\langle e_j e_k \right\rangle_w$ the $(j,k)$-entry of the inverse matrix $\mathcal{H}^{-1}
(w)$, up to a multiplicative factor $h_{N-1}$.
\end{proof}

A Toeplitz analogue of Corollary \ref{cor:GenHankel} is discussed in
Appendix \ref{app:Duduchava}. Both results showing that the 2-point kernel is a generating function are well known, and the Schur expansion here simply provides an immediate alternative proof.\par

\section{Kernels corresponding to classical ensembles}
\label{sec:ClassicCDK}

Our master formula \eqref{eq:kernelexp} implies that, to obtain the
coefficient in the Schur expansion of the kernel, we need to evaluate the
average $\langle s_{\lambda ^{\prime }}\rangle _{w}$ in the ensemble with
weight function $w(z)$. In this section we use known results for the
classical ensembles to finalize the computation of the coefficients.

An ensemble is said to be \emph{classical} if the weight function $w(z)$
satisfies Pearson's equation $\frac{\dd \ }{\dd z} \left[   \sigma (z) w (z)\right] =\tau (z) w (z)$ with $\sigma (z)$ and $\tau (z)$ polynomials with $\deg \sigma (z)\leq 2$ and $\deg \tau (z)=1$ \cite{Hild}. The list goes beyond the typical consideration as classical of many references, oftentimes limited to Gaussian, Laguerre and Jacobi ensembles. This is well-known in the orthogonal polynomials literature \cite{KM} and in the study of stationary solutions of stochastic processes \cite{Wong,FS}.

\subsection{Gaussian ensemble}

Consider the Gaussian unitary ensemble (GUE). The evaluation of the
coefficient $\langle s_{\lambda^{\prime}} \rangle_{\text{GUE}} $ relies on
the following result \cite{diFraItz}.

\begin{proposition}[Di Francesco--Itzykson \protect\cite{diFraItz}]
\label{prop:GUESchur}Let $\mu$ be a partition with $\ell (\mu)$ even, and
denote 
\begin{equation}  \label{eq:ljGUE}
l_j = \mu_j + \ell (\mu) -j , \qquad \forall j=1, \dots, \ell (\mu).
\end{equation}
Consider the GUE ensemble of $M$ variables. Then 
\begin{equation}
\langle s_{\mu} \rangle_{\mathrm{GUE}} = (-1)^{\frac{\ell (\mu)
\left( \ell (\mu) -2 \right) }{8}} \frac{ \prod_{\jmath :l_{\jmath}
\mathrm{odd}} l_{\jmath} !! \prod_{\tilde{\jmath}:l_{\tilde\jmath}\mathrm{even}} (l_{\tilde{\jmath}} -1) !! }{\prod_{\jmath:l_{\jmath}
\mathrm{odd}} \prod_{\tilde{\jmath}: l_{\tilde\jmath}\mathrm{even}}
(l_{\jmath} - l_{\tilde{\jmath}} )!! } ~ s_{\mu} (1^{M})
\end{equation}
if $\ell (\mu)$ is even, or 0 otherwise.
\end{proposition}

The coefficient in the Schur expansion of the kernel follows from the
specialization $\mu = \lambda^{\prime}$.

\subsection{Laguerre ensemble}

Consider the Laguerre unitary ensemble (LUE), with weight function $w(z)=
z^{\alpha} e^{-z} 1\!\!1_{z>0} $, $\alpha >-1$.

\begin{proposition}
\label{prop:LUESchur} Consider the LUE ensemble of $M$ variables and let $%
\mu $ be a partition with $\ell (\mu) \le M$. Then, 
\begin{equation}  \label{eq:LUESchur}
\langle s_{\mu} \rangle_{\mathrm{LUE}} = \prod_{j=1} ^{ M } \frac{
\Gamma \left( \alpha + \mu_j + M +1 -j \right) }{ \Gamma \left( \alpha + M
+1 -j \right) } ~ s_{\mu} (1^{M}).
\end{equation}
\end{proposition}

The evaluation of the coefficient in the Schur expansion follows from the
specialization $M= N-n$ and $\mu = \lambda^{\prime}$. For completeness, we
reproduce a proof of Proposition \ref{prop:LUESchur} due to \cite[Sec.4.2.2]%
{DavidGG} in Appendix \ref{app:SchurClassic}.

For the special case $\alpha \in \mathbb{N}_0$, \eqref{eq:LUESchur} can be
recast in equivalent forms: 
\begin{subequations}
\label{eq:LUESchurInteger}
\begin{align}
\langle s_{\mu} \rangle_{\mathrm{LUE} , \alpha
\in \mathbb{N}_0} & = \left[ \prod_{j=1} ^{M} (\alpha +M-j)^{-j} \right]
 \langle s_{\mu + (\alpha^M)} \rangle_{\mathrm{LUE},\alpha =0}  \label{eq:alphashiftLUE} \\
&= s_{\mu} (1^{M+\alpha}) \prod_{j=1} ^{M} \frac{ \Gamma \left( \mu_j -j +M
+1 \right) }{ \Gamma \left( M +1 -j \right) }  \label{eq:LUESchurInt1}
\end{align}
\end{subequations}
The first identity \eqref{eq:alphashiftLUE} is straightforward from \eqref{eq:LUESchur}, and can be alternatively derived from the integral representation absorbing the $\alpha$-dependent part of the Laguerre weight into the Schur polynomial, through the property 
\begin{equation*}
\left( \prod_{j=1} ^{M} z_j ^{\alpha} \right) s_{\mu} (z) = s_{\mu +
(\alpha^M)} (z) .
\end{equation*}
The $\alpha$-dependent but $\mu$-independent coefficient in \eqref{eq:alphashiftLUE} is entirely due to the denominator.

Using the dimension formula on $s_{\mu} (1^{M+ \alpha})$ and splitting the
product in the three regions: (i) $j<k \le M$, (ii) $j \le M$ with $k>M$ and
(iii) $M <j<k$ gives 
\begin{align*}
s_{\mu} (1^{M+ \alpha}) & = s_{\mu} (1^{M}) \prod_{j=1} ^{M} \prod_{k=M+1}
^{M+\alpha} \frac{ \mu_j - j +k }{k-j } \\
& =s_{\mu} (1^{M}) \prod_{j=1} ^{M} \frac{ \Gamma \left( \mu_j -j + M +
\alpha +1 \right) \Gamma \left( M +1-j \right) }{ \Gamma \left( \mu_j -j + M
+1 \right) \Gamma \left( M +\alpha +1-j \right) } ,
\end{align*}
using the hypothesis $\ell (\mu) \le M$. When plugged in \eqref{eq:LUESchur}, this proves \eqref{eq:LUESchurInt1}.\par
\medskip
Expressions for Schur averages closely related to \eqref{eq:LUESchur} at $\alpha=M$ have been given in \cite{Novaes:2015} (see also \cite{Nagar,Cunden} for related discussion), except that in \cite{Novaes:2015} the argument of $s_{\mu}$ in the average is inverted, $s_{\mu} (z_1^{-1} , \dots, z_M ^{-1})$. Recall that $s_{\mu} (1^M)$ gives the dimension of a $U(M)$ representation labelled by $\mu$, but the same partition can label a representation of the symmetric group $S_{\lvert \mu \rvert}$. Writing $s_{\mu} (1^M)$ as a function of the dimension of the $S_{\lvert \mu \rvert}$ representation recasts \eqref{eq:LUESchur} is a form very similar to \cite{Novaes:2015}.

\begin{remark}
In \cite[Eq.(96)]{Brown} the average $\langle s_{\mu} \rangle_{\mathrm{%
LUE }}$ at $\alpha=0$ is given as $(s_{\mu} (1^M))^2$, in conflict with our
computation. A check in \textsf{Mathematica} for $M=2,3,4$ and a small
sample of partitions $\mu$, however, shows agreement with formula
\eqref{eq:LUESchur}, while \cite[Eq.(96)]{Brown} fails. As further confirmation, the computations in Section \ref{sec:CBExp}, that rely on \eqref{eq:LUESchur}, are consistent with the existing literature.
\end{remark}

\subsection{Jacobi ensemble}
\label{sec:JUEavg}

Let us now take the (asymmetric) Jacobi ensemble with weight function $w(z)=
z^{\alpha} (1-z)^{\beta} 1\!\!1_{0<z<1} $, $\alpha >-1, \beta>-1$. The following holds, see e.g. \cite{Fyodorov:2006pk}.

\begin{proposition}
\label{prop:JUESchur} Consider the JUE ensemble of $M$ variables and let $%
\mu $ be a partition with $\ell (\mu) \le M$. Take, $\alpha, \beta \in 
\mathbb{N}_0$. Then, 
\begin{equation}  \label{eq:JUESchur}
\langle s_{\mu} \rangle_{\mathrm{JUE}} = \frac{ s_{\mu} (1^{M})
s_{\mu} (1^{\alpha + M}) }{s_{\mu} (1^{\alpha + \beta + 2 M})}.
\end{equation}
\end{proposition}

Formula \eqref{eq:JUESchur} can be equivalently written as 
\begin{equation}  \label{eq:JUESchurGamma}
\langle s_{\mu} \rangle_{\mathrm{JUE}} = \prod_{j=1} ^{M} \frac{
\Gamma \left( \mu_j -j + \alpha + M +1 \right) \Gamma \left( \alpha + \beta
+ 2M +1 -j\right) }{ \Gamma \left( \mu_j -j + \alpha + \beta + 2M +1 \right)
\Gamma \left( \alpha + M +1 -j\right) } ~ s_{\mu} (1^{M}) .
\end{equation}
The equality between \eqref{eq:JUESchur} and \eqref{eq:JUESchurGamma} is
shown using the dimension formula for the Schur polynomials. Many
simplifications take place separating the double products, both in the
numerator and in the denominator, in the three regions: (i) $j<k \le M$,
(ii) $j \le M$ and $k>M$, and (iii) $M < j <k$.

The average of the Schur polynomial in the form \eqref{eq:JUESchurGamma} can
be analytically continued to non-integer values of $\alpha, \beta$. Since
both sides of the equality \eqref{eq:JUESchurGamma}
depend analytically on $\alpha, \beta$ in a suitable region, we expect the
result \eqref{eq:JUESchurGamma} to hold when the restriction $\alpha, \beta
\in \mathbb{N}_0$ is lifted.

\subsubsection{Jacobi beta-ensemble}
As we have pointed out in Remark \ref{rmk:betaensemble}, the strategy can be applied to any beta-ensemble, with beta-parameter denoted by $2 \gamma$. It is convenient, however, to use the dual Cauchy identity \eqref{eq:DCJack} in terms of Jack polynomials. We are thus led to 
\begin{equation*}
    \left\langle \prod_{j=1}^{2n} \det \left( 1+ t_j Z \right) \right\rangle_{\mathrm{JE},2\gamma} = \sum_{\lambda \in \mathbb{Y}_{2n,M}} P^{(\gamma)} _{\lambda} (t) \left\langle P^{(1/\gamma)} _{\lambda^{\prime}} \right\rangle_{\mathrm{JE},2\gamma} 
\end{equation*}
where $\langle \cdot \rangle_{\mathrm{JE},2\gamma}$ is the average taken in the Jacobi ($2\gamma$)-ensemble. The advantage of this expansion is that the Jack polynomial average is known \cite{KadellJ}:
\begin{equation*}
    \left\langle P^{(1/\gamma)} _{\lambda^{\prime}} \right\rangle_{\mathrm{JE},2\gamma}  = \prod_{j=1} ^{2n} \frac{ (\alpha+1 + \gamma (M-j) )_{\lambda^{\prime} _j } }{  (\alpha+ \beta +2 + \gamma (2M-j-1) )_{\lambda^{\prime} _j} } \prod_{1 \le j < k \le \lambda_1 } \frac{ (\gamma(k-j+1))_{\lambda^{\prime} _j - \lambda^{\prime} _k} }{ (\gamma(k-j))_{\lambda^{\prime} _j - \lambda^{\prime} _k} } \prod_{j=1} ^{\lambda_1} \frac{ (\gamma(M-j+1))_{\lambda^{\prime} _j } }{ (\gamma(\lambda_1 -j +1))_{\lambda^{\prime} _j } } ,
\end{equation*}
where $(\cdot )_{n}$ is the Pochhammer symbol, and $\ell (\lambda^{\prime})= \lambda_1$ has been used.\par
Jack polynomial averages have been extensively used in \cite{FyoleDou,MeReWi} to compute (both positive and negative) moments in the Jacobi and Laguerre matrix beta-ensemble.

\subsection{Other classical ensembles}

The weight function $w(z)=\frac{z^{\alpha }}{(1+z)^{M+\beta }}%
1\!\!1_{0<z<\infty }$ also defines a classical ensemble since it satisfies Pearson's equation (see for example \cite{KM}). A Brownian motion
interpretation of this classical ensemble has been recently given in \cite[Remark 2.5]{BW}. The corresponding ensemble is related to a Jacobi ensemble,
in the sense that both weight functions are related to Euler's beta integral%
\begin{equation}
B(a,b)\equiv \int_{0}^{\infty }\frac{z^{a}}{\left( 1+z\right) ^{a+b}}\mathrm{%
d}z=\int_{0}^{1}u^{a-1}(1-u)^{b-1}\mathrm{d}z=\frac{\Gamma \left( a\right)
\Gamma \left( b\right) }{\Gamma \left( a+b\right) },  \label{Euler}
\end{equation}%
with $\Re a,\Re b\geq 0$. We denote such ensemble by $\widetilde{\mathrm{JUE}}$.

However, as pointed out in \cite{Fyodorov:2006pk}, whereas a change of
variables immediately relates the partition functions of the two ensembles, the same is not true if we have Schur insertions, since the argument of the Schur polynomial changes. Hence, the computation of the Schur average when $\alpha
,\beta \in \mathbb{N}_{0}$ is similar to but intrinsically different from the one for the Jacobi ensemble of Proposition \ref{prop:JUESchur} (cf. \cite{Fyodorov:2006pk}).

\begin{proposition}
\label{prop:RUESchur} Consider the $\widetilde{\text{JUE}}$ ensemble of $M$
variables and let $\mu $ be a partition with $\ell (\mu )\leq M$. Take, $%
\alpha ,\beta \in \mathbb{N}_{0}$ with $\tilde{\beta}\equiv \beta -\alpha
-M\geq 0$. Then, 
\begin{equation}
\langle s_{\mu }\rangle _{\widetilde{\mathrm{JUE}}}=\frac{s_{\mu
}(1^{M})s_{\mu }(1^{\alpha +M})}{s_{\mu ^{\prime }}(1^{\tilde{\beta}})}.
\label{eq:RUESchur}
\end{equation}
\end{proposition}\par
\medskip 
A further classical ensemble to be considered is the one with weight
function $w(z)=z^{-\tilde{\alpha}}e^{-1/z}1\!\!1_{0<z<\infty }$. It corresponds to the so-called inverse Wishart ensemble \cite{invWishart} and we schematically denote it as $\widetilde{\text{LUE}}$. 
Probabilistically, its weight function is the probability density function (pdf)
of the inverse Gamma distribution. Using the results for the Laguerre
ensemble we can solve this case as well, but the relationship between both
cases is non-trivial due to the presence of the Schur polynomial.

We rewrite 
\begin{equation}  \label{eq:alphaprimetoalpha}
\tilde{\alpha} = \alpha +2M-1
\end{equation}
and ask $\alpha>-1$, which is sufficient for convergence at $z \to \infty$.

A change of variables $z^{\prime} _j = z_j ^{-1}$ yields 
\begin{equation*}
\Delta_M (z)^2 \prod_{j=1} ^{M} \frac{e^{- z_j ^{-1}}}{z_j ^{\alpha +2N-1}} 
\mathrm{d} z_j = (-1)^M \Delta_M (z_j ^{\prime})^2 \prod_{j=1} ^{M} {z_j
^{\prime}} ^{\alpha} e^{ - z_j^{\prime} } \mathrm{d} z_j ^{\prime} .
\end{equation*}
(We henceforth drop the primes). In other words, the joint pdf of the $M$
eigenvalues in the $\widetilde{\text{LUE}}$ can be mapped to that of the LUE. The overall $(-1)^M$ will
cancel when computing averages.

Things become more involved for the average of a Schur polynomial, because
the change of variables inverts the argument: 
\begin{equation*}
\left\langle s_{\lambda ^{\prime }}\right\rangle _{\widetilde{\mathrm{LUE}}, \tilde{\alpha}}=\mathcal{N}_{M,\mathrm{LUE}}\int_{(0,\infty )^{M}}s_{\lambda ^{\prime
}}(z^{-1})\prod_{j=1}^{M}w_{\mathrm{LUE},\alpha}(z_{j})\mathrm{d}z_{j}.
\end{equation*}%
Besides, for the original integral to converge we must impose $\alpha >\ell
(\lambda )-1$. Recalling that this average appears in the sum %
\eqref{eq:kernelexp}, we impose $\alpha >2n-1$. We use the property \cite{MacDonald}\footnote{%
Analogous computations have been applied in \cite{Santilli:2018ilo,Santilli:2020snh} to study loop operators in 3d Chern--Simons-matter theories.} 
\begin{equation*}
s_{\lambda ^{\prime }}\left( \frac{1}{z_{1}},\dots ,\frac{1}{z_{M}}\right)
=s_{(\lambda ^{\prime })^{\ast }}\left( z_{1},\dots ,z_{M}\right)
~\prod_{j=1}^{M}z_{j}^{-\ell (\lambda )}
\end{equation*}%
where we have employed $\lambda _{1}^{\prime }=\ell (\lambda )$, and the partition $(\lambda ^{\prime })^{\ast }$ is
\begin{equation*}
(\lambda ^{\prime })^{\ast }=\left( \ell (\lambda )-\lambda _{M}^{\prime
},\ell (\lambda )-\lambda _{M-1}^{\prime },\dots ,0\right) .
\end{equation*}%
We conclude that%
\begin{equation}
\label{eq:SchurDLUEshift}
\left\langle s_{\lambda ^{\prime }}\right\rangle _{\widetilde{\mathrm{LUE}},\tilde{\alpha}}=\left\langle s_{(\lambda ^{\prime })^{\ast }}\right\rangle _{\mathrm{LUE},\alpha \mapsto \alpha -\ell (\lambda )},
\end{equation}
with the relation between $\tilde{\alpha}$ and $\alpha $ as in \eqref{eq:alphaprimetoalpha}. Since the Schur average in the LUE is known from \eqref{eq:LUESchur}, the Schur average in this other classical ensemble follows.\par
Precisely, the multi-point kernel of the inverse Wishart ensemble is studied in \cite{invWishart}, where it is interpreted as the probability distribution of non-interacting fermions in a background potential. Theorem \ref{thm:main} combined with \eqref{eq:SchurDLUEshift} and \eqref{eq:LUESchur} yields a character expansion of such quantity.


\begin{remark}
We have collected Schur averages for classical ensembles from the literature, with the exception of eventual rephrasing of some expressions and our result \eqref{eq:SchurDLUEshift}. In Section \ref{sec:qLaguerre}, we will compute the average in a $q$-Laguerre ensemble.
\end{remark}

\section[Laguerre Wronskian, conformal block expansion and tau function of Painlev\'{e} V]{Laguerre Wronskian, conformal block expansion and tau function of~Painlev\'{e}~V}
\label{sec:CBExp}

As an application of our results, we will exploit now the fact that the
diagonal of the kernels are directly related to Wronskians of orthogonal
polynomials. In turn, these Wronskians appear in webs of relationships that
oftentimes include Painlev\'{e} tau functions. We show here how the diagonal limit of the Laguerre expansion obtained before is related to the random matrix results in \cite{Winn,MLA}.

The diagonal limit of the Christoffel--Darboux~kernel, when all variables are the same, 
\begin{equation*}
K_N ^{(n)} (\underbrace{x, \dots, x}_{n}; \underbrace{x, \dots, x}_{n} )
\equiv \lim_{(x_1, \dots, x_n) \to (x, \dots, x)} \lim_{ (y_1, \dots, y_n)
\to (x_1, \dots, x_n) } K_N ^{(n)} (x_1, \dots , x_n ; y_1 , \dots , y_n )
\end{equation*}
corresponds to the moment of order $2n$ of the characteristic polynomial
average in the corresponding random matrix ensemble: 
\begin{equation}  \label{eq:BHIntegral}
\frac{1}{\prod_{j=N-n} ^{N-1} h_j} K_N ^{(n)} (x, \dots, x; x, \dots, x) =
\left\langle \det( x - Z)^{2n} \right\rangle_{w} ,
\end{equation}
where $Z$ is a random matrix from the ensemble with weight function $w(z)$ in $N-n$ variables. The random matrix average \eqref{eq:BHIntegral} is a $2n
\times 2n $ Wronskian determinant built with the orthogonal polynomials $P_k
(z)$ associated to $w(z)$ \cite[Eq.(15)]{BH} 
\begin{equation}  \label{eq:BHWronskian}
\left\langle \det( x - Z)^{2n} \right\rangle_{w} = \frac{1}{G(2n+1)} \mathrm{%
Wr} \left( P_{N-n} (x) , P_{N+1-n} (x) , \dots, P_{N+n-1} (x) \right).
\end{equation}
Typically, the Wronskian is equivalently written as a Hankel determinant
involving \emph{different} orthogonal polynomials \cite{KSz,Leclerc}.

In the case of the Laguerre ensemble, the average \eqref{eq:BHIntegral}
plays a prominent role and is related to a wide range of quantities \cite{Forr94}. The Hankel equivalent was studied in detail in \cite{Winn}, and later on in \cite{MLA}.

It is convenient for the rest of the subsection to adopt the shorthand%
\begin{equation*}
\widetilde{K}_{N}^{(n)}(x)_{\mathrm{LUE}}\equiv \frac{1}{%
\prod_{j=N-n}^{N-1}h_{j}}K_{N}^{(n)}(\underbrace{x,\dots ,x}_{n};\underbrace{%
x,\dots ,x}_{n})|_{\mathrm{LUE},\alpha =2n}.
\end{equation*}%
Then%
\begin{align*}
\widetilde{K}_{N}^{(n)}(x)_{\mathrm{LUE}}& \overset{\text{%
\eqref{eq:BHIntegral}}}{=}\left\langle \det (x-Z)^{2n}\right\rangle _{\text{%
\textrm{LUE}},\alpha =2n} \\
& \overset{\text{\eqref{eq:BHWronskian}}}{=}(-1)^{n}\frac{G(M+2n+1)}{%
G(2n+1)G(M+2)}~\mathrm{Wr}\left( L_{N-n}^{(\alpha
=2n)}(x),L_{N+1-n}^{(\alpha =2n)}(x),\dots ,L_{N+n-1}^{(\alpha
=2n)}(x)\right) ,
\end{align*}%
where $L_{k}^{(\alpha )}$ are the generalized Laguerre polynomials and the additional Barnes $G$-functions \cite{Forrester} on the second line come from passing from
monic to orthonormal polynomials. Note that we are pegging the value of the
parameter $\alpha $ with the order of the moment of the characteristic
polynomial, $\alpha =2n$.

Moreover, following \cite{MLA}, let us introduce the function 
\begin{align}
f_{2n} (x) &= e^{- \frac{M}{2} x } (-1)^n ~ \mathrm{Wr} \left( L_{N-n}
^{(\alpha =2n)} (- x) , L_{N+1-n} ^{(\alpha =2n)} (- x) , \dots, L_{N+n-1}
^{(\alpha =2n)} (- x) \right)  \notag \\
&= \frac{ G(2n+1) G(M+2) }{G(M+2n+1)} ~e^{- \frac{M}{2} x } ~ \widetilde{K}%
_N ^{(n)} (- x)_{\mathrm{LUE }} .  \label{eq:F2ntoLUE}
\end{align}
Then, combining the above discussion with \cite[Prop.3]{Winn} yields another
expression for the Laguerre kernel.

\begin{proposition}[Winn \cite{Winn}]
\label{prop:Winn1} Let $\zeta \in \mathbb{R}$ and set $x=2 \lvert \zeta
\rvert$. Consider the LUE of $M=N-n$ variables with $\alpha = 2n$. The
associated kernel satisfies the identity 
\begin{equation*}
f_{2n} (x) = \frac{2^{M^2+4Mn}}{ (2\pi)^M M! } \int_{\mathbb{R}^{M}}
\Delta_{M} (z)^2 \prod_{j=1} ^{M} \frac{ e^{\mathrm{i} \xi z_j } }{ (1+z_j
^2 )^{M+2n} } \mathrm{d} z_j ,
\end{equation*}
with $f_{2n}$ related to the LUE kernel via \eqref{eq:F2ntoLUE}.
\end{proposition}

Proposition \ref{prop:Winn1} is just a rephrasing of \cite[Prop.4]{Winn},
but here we are emphasizing the interpretation as a kernel for the LUE.

The Wronskian of Laguerre polynomials is in turn related to a solution to
the $\sigma$-Painlev\'{e} V equation \cite[Thm.1]{MLA}. The result of interest for us is that the function $f_{2n} (x)$ is essentially a tau function of Painlev\'{e} V. This implies a connection between the kernel of the LUE, with $\alpha=2n$, and tau functions of Painlev\'{e} V. The mapping between the two settings is non-linear, and indeed requires the term  $e^{-\frac{M}{2} x}$.

Combining the present discussion with the Schur expansion, using the explicit knowledge of the coefficient from \eqref{eq:LUESchurInt1} with $\alpha =2n$, we derive a new expansion for $f_{2n}(x)$: 
\begin{equation*}
f_{2n}(x)=\frac{G(2n+1)}{G(M+2n+1)}e^{-\frac{M}{2}x}x^{2Mn}\sum_{\lambda \in 
\mathbb{Y}_{2n,M}}x^{-\lvert \lambda \rvert }s_{\lambda }(1^{2n})s_{\lambda
^{\prime }}(1^{M+2n})\prod_{j=1}^{M}\Gamma \left( {\lambda ^{\prime }}%
_{j}-j+M+1\right) .
\end{equation*}%
The sum runs over partitions contained in a $2n\times M$ rectangle, thus the
overall factor $x^{2Mn}$ guarantees that only non-negative powers of $x$ are
included. Explicitly, expanding the exponential and rearranging the terms,
we find%
\begin{align}
f_{2n}(x)=\frac{G(2n+1)}{G(M+2n+1)}& \sum_{k=0}^{\infty }x^{k}\left\{
\sum_{m=0}^{k}\sum_{\substack{ \mu \in \mathbb{Y}_{2n,M} \\ \lvert \mu
\rvert =m}}\left( -\frac{M}{2}\right) ^{k-m}\frac{1}{(k-m)!}\right.   \notag
\\
& \hspace{-12pt} \times \left. s_{(M^{2n})\setminus \mu }(1^{2n})s_{((2n)^{M})\setminus \mu
^{\prime }}(1^{2n+M})~\prod_{j=1}^{M}\Gamma \left( 2n+M+1-j-\mu
_{M-j+1}^{\prime }\right) \right\} .  \label{eq:F2nSchurExp}
\end{align}%
We have replaced the sum over $\lambda $ by a sum over partitions $\mu $
such that $\lambda $ is realised subtracting $\mu $, rotated by $180^{\circ }
$, from the bottom-right corner of the $2n\times M$ rectangular partition $%
(M^{2n})$. Then, $\lambda ^{\prime }$ is realised by subtracting $\mu
^{\prime }$ to the rectangular partition $((2n)^{M})$ in the same way. For
example, for $M=6$ and $2n=4$, 
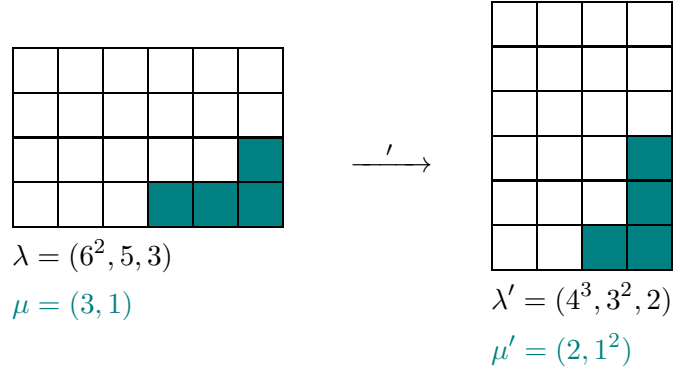
\begin{equation*}
\begin{aligned} & \begin{ytableau} \ & \ & \ & \ & \ & \  \\ \ & \ & \ & \ &
\ & \  \\ \ & \ & \ & \ & \ & *(teal) \ \\ \ & \ & \ & *(teal) \ & *(teal) \
& *(teal) \ \end{ytableau}\\ & \lambda = (6^2,5,3) \\ &
\begin{color}{teal}\mu = (3,1)\end{color} \end{aligned}\qquad %
\begin{aligned} & \xrightarrow{ \quad \prime \quad } \\ &\hspace{1cm}
\end{aligned}\qquad \begin{aligned} & \begin{ytableau} \ & \ & \ & \  \\ \ &
\ & \ & \  \\ \ & \ & \ & \  \\ \ & \ & \ & *(teal) \ \\ \ & \ & \ & *(teal)
\ \\ \ & \ & *(teal) \ & *(teal) \  \end{ytableau}\\ & \lambda^{\prime}
=(4^3,3^2,2) \\ & \begin{color}{teal}\mu^{\prime} = (2,1^2)\end{color}
\end{aligned}
\end{equation*}

In particular, writing as in \cite{MLA} 
\begin{equation*}
f_{2n} (x) = f_{2n} (0) \left[ 1 + b_1 x + b_2 x^2 + O(x^3) \right] ,
\end{equation*}
from \eqref{eq:F2nSchurExp} we find 
\begin{align*}
f_{2n} (0) &= s_{((2n)^{M})} (1^{M+2n}) \\
&= G(M+1) \frac{ G(4n+M+1) G(2n+1)^2 }{ G(4n+1) G(2n+M+1)^2 }
\end{align*}
and 
\begin{align*}
b_1 &= \begin{color}{teal}\overbrace{\begin{color}{black}- \frac{M}{2}
\end{color}}^{\mu = \emptyset} \end{color} + \begin{color}{teal} \overbrace{
\begin{color}{black} \frac{1}{2n} \underbrace{ \frac{ s_{(M^{2n}) \setminus
(1) } (1^{2n}) s_{((2n)^M) \setminus (1) } (1^{M+2n}) }{ s_{((2n)^{M})}
(1^{M+2n}) } }_{ = Mn } \end{color}}^{\mu = (1)} \end{color} \ = 0 ,
\end{align*}
\begin{align*}
b_2 &= \begin{color}{teal}\overbrace{\begin{color}{black} \frac{M^2}{8}
\end{color}}^{\mu = \emptyset} \end{color} \begin{color}{teal}\overbrace{%
\begin{color}{black} - \frac{M}{2} \cdot \frac{1}{2n} \underbrace{ \frac{
s_{(M^{2n}) \setminus (1) } (1^{2n}) s_{((2n)^M) \setminus (1) } (1^{M+2n})
}{ s_{(2n)^{M}} (1^{M+2n}) } }_{ = Mn } \end{color}}^{\mu = (1)} \end{color}
\\
& \begin{color}{teal}\overbrace{\begin{color}{black} + \frac{1}{2n (2n+1)}
\underbrace{ \frac{ s_{(M^{2n}) \setminus (2) } (1^{2n}) s_{((2n)^M)
\setminus (1^2) } (1^{M+2n}) }{ s_{(2n)^{M}} (1^{M+2n}) } }_{ = M(M-1)
\frac{n(2n+1)^2}{4(4n+1)} } \end{color}}^{\mu = (2)} \end{color} \\
& \begin{color}{teal}\overbrace{\begin{color}{black} + \frac{1}{2n (2n-1)}
\underbrace{ \frac{ s_{(M^{2n}) \setminus (1^2) } (1^{2n}) s_{((2n)^M)
\setminus (2) } (1^{M+2n}) }{ s_{(2n)^{M}} (1^{M+2n}) } }_{ = M(M+1)
\frac{n(2n-1)^2}{4(4n-1)} } \end{color}}^{\mu = (1^2)} \end{color} \ = - 
\frac{M(M+4n)}{ 8(4n+1)(4n-1) }
\end{align*}
in perfect agreement with \cite{MLA}.\par
\medskip
The tau function of Painlev\'{e} V, and thus $f_{2n}$, admits an
expansion in conformal blocks \cite{LisovyyCB}, involving a differently organized sum over partitions. Comparing our formula with the conformal
block expansion, we make the following observations.

\begin{itemize}
\item In the spirit of Theorem \ref{thm:main}, we get a sum over a single
partition, as opposed to the sum over pairs of partitions in the conformal
block expansion.

\item The Wronskian presentation \eqref{eq:BHWronskian} makes it manifest
that the dependence on $x$ must be a polynomial times the exponential factor 
$e^{- \frac{M}{2} x}$. This property is explicit in the Schur expansion %
\eqref{eq:F2nSchurExp}, as opposed to the conformal block expansion. Indeed,
the sum over partitions in \eqref{eq:F2nSchurExp} terminates.

\item The conformal block expansion makes the dependence on the parameters
of the Painlev\'{e} equation more manifest.
\end{itemize}

In summary, our expansion is meant for the pure Wronskian part, whereas the
conformal block and related expansions \cite{LisovyyCB,MLA} are suitable for
Painlev\'{e} transcendents and tau functions.

\section{Kernels corresponding to $q$-ensembles}
\label{sec:qdefCDK}

Throughout the present section we compute the coefficients in the Schur
expansion of the kernel for a selected class of $q$-ensembles. The
definition of $q$-ensemble is the one put forward in \cite{MCI}, namely, a
standard random matrix ensemble whose weight function is such that its
associated orthogonal polynomials are $q$-deformed. We will focus on the
Stieltjes--Wigert ensemble, which has multitude of physical applications 
\cite{ForresterCrystal,Tierz:2002jj,Tierz:2017nvl,Santilli:2018ilo,ForrSW},
and its one-parameter generalization, the $q$-Laguerre ensemble, which of
course is also a one-parameter generalization of the Laguerre ensemble
studied above.

\subsection{Stieltjes--Wigert ensemble}

Consider the Stieltjes--Wigert~ensemble, whose weight function is $w(z)= 
\frac{1}{\sqrt{2 \pi g}} e^{- \frac{1}{2g} (\ln z)^2} 1\!\!1_{0<z<\infty}$, $%
g>0$. It is customary to introduce the $q$-parameter $q= e^{-g}$. To set the
notation, we define the symmetric $q$-number as 
\begin{equation}  \label{eq:defqnumber}
\left[ z \right]_{q} = \frac{q^{-z/2} - q^{z/2}}{q^{- 1/2} - q^{1/2}} .
\end{equation}
The $q$-dimension of a $U(M)$ representation labelled by a partition $\mu$
is 
\begin{equation}  \label{eq:defqdim}
\dim_q \mu = \prod_{1 \le j < k \le M } \frac{ \left[ \mu_j -j - \mu_k +k %
\right]_q }{ \left[ j-k \right]_q } .
\end{equation}

\begin{proposition}
\label{prop:SWSchur} Consider the Stieltjes--Wigert~ensemble of $M$
variables and let $\mu$ be a partition with $\ell (\mu) \le M$. Then 
\begin{equation}  \label{eq:SWSchur}
\langle s_{\mu} \rangle_{\mathrm{SW }} = q^{ - \frac{1}{2}
\sum_{j=1}^{M} \mu_j \left( \mu_j + 3 M +1 - 2 j \right) }~ \dim_q \mu ,
\end{equation}
\end{proposition}

This result already appears in \cite{Dolivet:2006ii}, and we give an
alternative proof in Appendix \ref{app:SchurqSW}.\par
\medskip 
It is worth mentioning that, as for the Laguerre ensemble discussed in Section \ref{sec:CBExp}, the diagonal limit of the
multidimensional Stieltjes--Wigert~kernel can be related to other objects
previously studied in the literature. In this case, with the partition function of a model of fermions with a large
non-Abelian symmetry \cite{Anninos:2016klf,Tierz:2017nvl,DGMT}. Indeed, the diagonal limit
of the kernel is proportional to the Wronskian of Stieltjes--Wigert polynomials, which was shown in \cite{Tierz:2017nvl} to compute the partition function $\mathcal{Z}_{M,2n}^{\mathrm{fermion}}(x)$ of the fermion model. Therefore, we have the character expansion 
\begin{align*}
\mathcal{Z}_{M,2n}^{\mathrm{fermion}}(x)& =\frac{1}{M!}%
\int_{(0,\infty )^{M}}\Delta (z)^{2}\prod_{j=1}^{M}\left( x-z_{j}\right)
^{2n}e^{-\frac{1}{2g}(\log z_{j})^{2}}\frac{\mathrm{d}z_{j}}{\sqrt{2\pi g}}
\\
& =\frac{\prod_{j=0}^{N-n-1}h_{j}}{\prod_{j=N-L}^{N-1}h_{j}}K_{N}^{(n)}(%
\underbrace{x,\dots ,x}_{n};\underbrace{x,\dots ,x}_{n})|_{\mathrm{SW}} \\
& =\mathcal{Z}_{M} ^{\mathrm{SW}} (q) ~ \sum_{\lambda \in \mathbb{Y}_{2n,M}}x^{2nM-\lvert \lambda \rvert
} ~ q^{-\frac{3M+1}{2} \lvert \lambda \rvert + \sum_{j=1}^{M} \left( - \frac{1}{2} (\lambda_j ^{\prime} )^2 + j \lambda ^{\prime} _j \right)   } ~ \dim \lambda  \dim_q \lambda ^{\prime },
\end{align*}%
where $x$ is a spectral parameter in the theory \cite{Tierz:2017nvl,DGMT}, $\dim \lambda = s_{\lambda} (1^{2n})$ is the dimension of the $U(2n)$ representation labelled by $\lambda$ and 
\begin{equation*}
\mathcal{Z}_{M} ^{\mathrm{SW}} (q) =\left[ \prod_{j=1}^{M-1}\Gamma _{q}\left( 1+j\right) \right] \left(
1-q\right) ^{\frac{M(M-1)}{2}}q^{-\frac{M(M^{2}-1)}{6}} 
\end{equation*}
is the Stieltjes--Wigert partition function. In this way, the setup of the present work provides a new interpretation of the physically meaningful quantity $\mathcal{Z}_{M,2n}^{\mathrm{fermion}}(x)$, together with its character expansion, in the spirit of \cite{Santilli:2020snh}. Being the Schur expansion of $\mathcal{Z}_{M,2n}^{\mathrm{fermion}}(x)$ a generating polynomial in the fugacity $x$, the expression is especially well-suited for comparison with analogous observables in distinct theories, for instance with the aim of testing dualities.\par
\medskip
Next we discuss the case of the $q$-Laguerre ensemble, which has not been studied before in this context and has the interesting feature that it
generalizes the Stieltjes--Wigert ensemble (a $q$-ensemble) while also being
a $q$-deformation of the classical Laguerre ensemble. We will be checking
out these generalizing features of the ensemble, by taking the appropriate limits.

\subsection{$q$-Laguerre ensemble}
\label{sec:qLaguerre}

Consider the $q$-Laguerre ensemble, with weight function $w(z)= \frac{%
z^{\alpha}}{ (-(1-q)z;q)_{\infty}}$, with $(\cdot ;q)_{\infty}$ the $q$-Pochhammer symbol and $0<q<1$. The corresponding orthogonal polynomials
have been constructed by Moak \cite{Moak}, thus the kernel can be obtained
explicitly from the Christoffel--Darboux~formula.

Here we instead compute the average of a Schur polynomial in the $q$-Laguerre ensemble, providing a $q$-deformation of the result in Proposition %
\ref{prop:LUESchur}. Note that our definition of $q$-Laguerre weight is as
in \cite{Moak}, and differs from \cite[Sec.14.2]{Koekoek} and 
\cite{Christiansen1} by a normalization factor, as in these references the denominator is $(-z;q)_{\infty }$.

\begin{proposition}
\label{prop:qLagSchur} Consider the $q$-Laguerre ensemble of $M$ variables
and let $\mu$ be a partition with $\ell (\mu) \le M$. Then 
\begin{align}  \label{eq:qLagSchur}
\langle s_{\mu} \rangle_{q\mathrm{LUE}} & = q^{ - \frac{1}{2}
\sum_{j=1}^{M} \mu_j \left( M -1 \right) } ~ \dim_q \mu \\
& \times \prod_{j=1} ^{M} \frac{ \Gamma \left( \alpha +1 + \mu_j + M -j
\right) \Gamma (- \alpha - \mu_j - M +j)}{ \Gamma (\alpha +j ) \Gamma (-
\alpha -j +1) } \frac{ \Gamma_q (- \alpha -j +1) }{ \Gamma_q (- \alpha -
\mu_j - M +j) } .  \notag
\end{align}
\end{proposition}

The result \eqref{eq:qLagSchur} is new, to our knowledge, and we give a
proof in Appendix \ref{app:SchurqLUE}. It is straightforward to see that %
\eqref{eq:qLagSchur} converges to \eqref{eq:LUESchur} in the $q \to 1$ limit.
 io 
For $\alpha \notin \mathbb{Z}$ the expression \eqref{eq:qLagSchur} may be
simplified using Euler's reflection property $\Gamma (1-z) \Gamma (z) = 
\frac{\pi}{\sin (\pi z)}$, $z \notin \mathbb{Z}$. In turn, for $\alpha \in 
\mathbb{N}_0$, property \eqref{eq:alphashiftLUE} carries over directly to
the $q$-deformed setting.\par
\medskip
It has been shown by Askey \cite{Askey1} that the $q$-Laguerre
polynomials, upon changing the weight to $w(z)= \frac{z^{\alpha}}{
(-z;q)_{\infty}}$ and scaling the variable $z \mapsto q^{-\alpha} z$,
converge to the Stieltjes--Wigert~polynomials in the $\alpha \to \infty$
limit. We now discuss the implications for the kernel as seen from the Schur
expansion.

\begin{enumerate}[(i)]

\item \label{AskeyStep1} First, we pass from Moak's \cite{Moak} to Askey's 
\cite{Askey1} normalization. We can either

\begin{itemize}
\item redo the computations in Appendix \ref{app:SchurqLUE} with modified
moments, or

\item we can simply notice that, when computing the average $\langle s_{\mu}
\rangle $ all prefactors will cancel against the normalization except for
the one coming from a rescaling of the variables in $s_{\mu} (z)$.
\end{itemize}
In both ways we find that we must include a term $(1-q)^{\lvert \mu \rvert}$
in the average of the Schur polynomial.

\item We then take the large $\alpha$ limit of the $\alpha$-dependent but $q$%
-independent part in \eqref{eq:qLagSchur}. It is convenient to take the
limit $\alpha \to \infty$ with $\alpha \notin \mathbb{Z}$, in which case 
\begin{align}
\lim_{\alpha \to \infty} \prod_{j=1} ^{M} \frac{ \Gamma \left( \alpha +1 +
\mu_j + M -j \right) \Gamma (- \alpha - \mu_j - M +j)}{ \Gamma (\alpha +j )
\Gamma (- \alpha -j +1) } & = \lim_{\alpha \to \infty} \frac{ (-1) \sin \pi
(\alpha + j) }{ \sin \pi (\alpha + \mu_j -j + M) }  \notag \\
& = \prod_{j=1} ^{M} (-1)^{\mu_j + M -1} = (-1)^{\lvert \mu \rvert}
\label{eq:signfromGammaration}
\end{align}

\item We use basic properties of the $\Gamma_q$ function to rewrite %
\eqref{eq:qLagSchur} in a more suitable form: 
\begin{equation}  \label{eq:rewritingQGamma}
\prod_{j=1} ^{M} \frac{ \Gamma_q (- \alpha -j +1) }{ \Gamma_q (- \alpha -
\mu_j - M +j) } = (-q)^{- \lvert \mu \rvert} \prod_{j=1} ^{M} \frac{
\Gamma_{q^{-1}} (\alpha + \mu_j -j + M ) }{ \Gamma_{q^{-1}} ( \alpha +j -1) }
.
\end{equation}

\item To take $\alpha \to \infty$ in the product in %
\eqref{eq:rewritingQGamma}, we use Moak's $q$-analogue of Stirling's formula 
\cite{Moak2}. It gives 
\begin{align}
\prod_{j=1} ^{M} \frac{ \Gamma_{q^{-1}} (\alpha + \mu_j -j + M ) }{
\Gamma_{q^{-1}} ( \alpha +j -1) } & \approx (1-q^{-1})^{-\lvert \mu \rvert }
\prod_{j=1} ^{M} \frac{ (-q) ^{- \left[ (\alpha + \mu_j -j +M)(\alpha +
\mu_j -j +M- \frac{1}{2} ) - \frac{1}{2} (\alpha + \mu_j -j +M)^2 \right] } 
}{ (-q) ^{- \left[ (\alpha + j -1)(\alpha + j - \frac{3}{2} ) - \frac{1}{2}
(\alpha + j -1)^2\right] } }  \notag \\
&= (-1)^{\lvert \mu \rvert} (1-q^{-1})^{-\lvert \mu \rvert } q^{- \frac{1}{2}
\sum_{j=1} ^{M} \mu_j \left( \mu_j +2M -2j -1 \right) } q^{ - \alpha \lvert
\mu \rvert} .
\end{align}
The overall sign $(-1)^{\lvert \mu \rvert} $ cancels against %
\eqref{eq:signfromGammaration}, while combining the factor $%
(1-q^{-1})^{-\lvert \mu \rvert }$ with $(-q)^{\lvert \mu \rvert}$ from %
\eqref{eq:rewritingQGamma} cancels the factor from the change of weight (cf.
step \eqref{AskeyStep1}).

\item Putting all the pieces together and comparing with \eqref{eq:SWSchur},
we arrive at 
\begin{equation}  \label{eq:LUESchurToSWSchur}
\lim_{\alpha \to \infty} q^{(\alpha -1)\lvert \mu \rvert} \left\langle s_{\mu} \right\rangle_{q\mathrm{LUE}} =
 \left\langle s_{\mu} \right\rangle_{\text{\textrm{SW}}} .
\end{equation}

\item So far we have discussed the $\alpha \to \infty$ limit of the Schur
average alone. We additionally impose the scaling $(x_1, \dots, x_n; y_1,
\dots, y_n) \mapsto ( q^{1-\alpha} x_1, \dots, q^{1-\alpha} x_n;
q^{1-\alpha} y_1, \dots, q^{1-\alpha} y_n)$, so that a factor $q^{(\alpha-1)
\lvert \lambda \rvert}$ comes from the Schur polynomial, cancelling the
prefactor in \eqref{eq:LUESchurToSWSchur}.\footnote{%
We need to take the scaling $x \mapsto q^{1-\alpha } x$, instead of Askey's $%
x \mapsto q^{-\alpha } x$. This seems to be due to the fact that Askey
showed convergence to a \emph{different} weight with same associated system
of orthogonal polynomials.}
\end{enumerate}

\section{Kernels on $\mathbb{C}$ and generalizations of Dotsenko--Fateev
integrals}
\label{sec:CKernels}

The Schur expansion can be applied to reproducing kernels on $\mathbb{C}$,
such as Bergman kernels and kernels on Bargmann--Fock spaces. Because kernels
in a complex space setting are, in quite a few instances, of a very simple
nature, in those cases the type of expansion discussed above will not yield
an alternative expression. We show this explicitly with the
straightforward case of the Ginibre ensemble. Even in more complicated
scenarios, as for example with polyanalytic Bargmann--Fock spaces \cite%
{poly1,poly2}, in some cases (but certainly not all) the kernel itself may be as simple as a single Laguerre polynomial \cite{poly1,poly2}.\par
For non-trivial results, we focus here on other models, and we discuss thoroughly a Dotsenko--Fateev ensemble as our main example.\par
\medskip
The starting point of our analysis is the representation of the kernel as the average of a
product of characteristic polynomials, that generalizes the real case. We
will need a particular case of a theorem by Akemann and Vernizzi \cite{AV}.

To set the stage, let $(x_1, \dots ,x_n) \in \mathbb{C}^n$ and $(y_1, \dots
,y_n) \in \mathbb{C}^n$, $n \ge 1$. The $x_i$ will be related to the
holomorphic sector and the $y_j$ to the anti-holomorphic sector. As a
consequence, the kernel will depend holomorphically on the $x_i$ and
anti-holomorphically on $y_j$. The notation \eqref{eq:defmultikernel} is
then extended to the complex case as 
\begin{equation*}
K_{N} ^{(n)} (x_1, \dots, x_n; \bar{y}_1, \dots, \bar{y}_n ) = \frac{1}{%
\Delta_n (x) \Delta_n (\bar{y}) } \det_{1 \le i,j \le n} \left[ K_N (x_i, 
\bar{y}_j) \right] .
\end{equation*}
The extension of the integral representation \eqref{eq:KNint} to complex
ensembles is as follows \cite{AV}.

\begin{proposition}[Akemann--Vernizzi \protect\cite{AV}]
With the notation as in Section \ref{sec:GeneralSchurExp}, 
\begin{align*}
K_{N} ^{(n)} (x_1, \dots, x_n; \bar{y}_1, \dots, \bar{y}_n ) = \frac{ 
\mathcal{N}_{M} }{ \prod_{j= M}^{N-1} h_j } \int \Delta_{M} (z) \Delta_{M} (%
\bar{z} ) \prod_{j=1} ^{M} \left[ \prod_{i=1}^{n} \left( x_i -z_j \right)
\left( \bar{y}_i - \bar{z}_j \right) \right] w (z_j , \bar{z}_j) \mathrm{d}
z_j \mathrm{d} \bar{z}_j .
\end{align*}
\end{proposition}

As opposed to the real case, the Schur expansion of the kernel on $\mathbb{C}
$ will require two Schur polynomials: one for the holomorphic and one for
the anti-holomorphic sector. Therefore, the Schur expansion is akin to
Theorem \ref{thm:double}. Applying the dual Cauchy identity %
\eqref{eq:dualCauchy} twice we get 
\begin{equation*}
\widehat{K}_N ^{(n)} (x_1, \dots, x_n; \bar{y}_1, \dots, \bar{y}_n ) =
\sum_{\lambda, \mu \in \mathbb{Y}_{n,M}} s_{\lambda} \left( x^{\vee} \right)
s_{\mu} \left( \bar{y}^{\vee} \right) \left\langle s_{\lambda^{\prime}} \bar{%
s}_{\mu^{\prime}} \right\rangle_w ,
\end{equation*}
where 
\begin{equation*}
\left\langle s_{\lambda^{\prime}} \bar{s}_{\mu^{\prime}} \right\rangle_w
\equiv \mathcal{N}_{M} \int \Delta_{M} (z) \Delta_{M} (\bar{z})
s_{\lambda^{\prime}} (z) s_{\mu^{\prime}} (\bar{z}) \prod_{j=1}^{N-n} w(z_j, 
\bar{z}_j) \mathrm{d} z_j \mathrm{d} \bar{z}_j .
\end{equation*}
Notice the lack of symmetry enhancement \eqref{eq:symen} in the complex
case, due to one set of variables appearing holomorphically and the other
anti-holomorphically.

\subsection{Ginibre ensemble}

The first example we consider is the complex Ginibre ensemble. The
corresponding monic orthogonal polynomials on $\mathbb{C}$ are the
monomials, $P_j (z) = z^{j}$ with $h_j = j!$. Therefore, with the notation %
\eqref{eq:defhatK}, 
\begin{equation}  \label{eq:kernelGinibre}
\widehat{K}_N (x; \bar{y} ) \vert_{\mathrm{Gin}} = (N-1)!
\sum_{j=0}^{N-1} \frac{ (x \bar{y})^{j-N +1}}{j!} .
\end{equation}
To compare with its expansion in the basis of Schur polynomials, we use \cite[Eq.(3.16)]{ForresterRains} 
\begin{equation*}
\left\langle s_{\lambda^{\prime}} \bar{s}_{\mu^{\prime}} \right\rangle_{\mathrm{Gin}} = \delta_{\lambda^{\prime} \mu^{\prime}} \prod_{j=1}
^{M} \frac{ \Gamma \left( M -j +1 + \lambda^{\prime} _j \right) }{ \Gamma
\left( M -j +1 \right) } .
\end{equation*}
The double sum reduces to a single sum, 
\begin{equation*}
\widehat{K}_N ^{(n)} (x_1, \dots, x_n; \bar{y}_1, \dots, \bar{y}_n ) \vert_{%
\mathrm{Gin}} = \sum_{\lambda \in \mathbb{Y}_{n, N-n} } s_{\lambda}
\left(x^{\vee} \right) s_{\lambda} \left( \bar{y}^{\vee} \right)
\left\langle s_{\lambda^{\prime}} \bar{s}_{\lambda^{\prime}} \right\rangle_{%
\mathrm{Gin}} .
\end{equation*}
For the particular case $n=1$, using 
\begin{equation*}
\lambda^{\prime} _j= 
\begin{cases}
1 & 1 \le j \le \lambda_1 \\ 
0 & \text{otherwise}%
\end{cases}
\quad \Longrightarrow \quad \prod_{j=1} ^{M} \frac{ \Gamma \left( N -j +
\lambda^{\prime} _j \right) }{ \Gamma \left(N-j \right) } = \prod_{j=1}
^{\lambda_1} (N-j) ,
\end{equation*}
the agreement with \eqref{eq:kernelGinibre} is immediately checked.\par
\medskip
Based on Remark \ref{rmk:betaensemble}, we can study the kernel in
the real Ginibre ensemble. We use 
\begin{equation*}
K_N  (x;y) \vert_{\mathbb{R}\mathrm{-Gin}} = (x-y) \left\langle
\det (x-M) \det (y-M^{T}) \right\rangle_{\mathbb{R}\mathrm{-Gin}}
\end{equation*}
where $M$ is a random matrix taken from the real Ginibre ensemble.
Proceeding as for the complex Ginibre ensemble, but this time using the
expansion from Theorem \ref{thm:main} instead of Theorem \ref{thm:double},
we independently reproduce the result of \cite{APS}: 
\begin{equation*}
K_N  (x;y) \vert_{\mathbb{R}\mathrm{-Gin}} = (x-y) (N-1)!
\sum_{j=0} ^{N-1} \frac{(xy)^{j}}{j!} .
\end{equation*}

\subsection{Dotsenko--Fateev ensemble}

The coefficient in the Schur expansion can be evaluated explicitly for the
weight function 
\begin{equation*}
w_{\text{DF}}(z,\bar{z})=w_{\text{JUE}}(z)w_{\text{JUE}}(\bar{z}%
)=|z|^{2\alpha }|1-z|^{2\beta }.
\end{equation*}%
The corresponding partition function is a Dotsenko--Fateev integral \cite%
{dotsenko}, which is a complex version of the Selberg integral, 
\begin{align*}
\mathcal{Z}_{M}^{\text{DF}}& =\frac{1}{M!}\int \left[ \Delta _{M}(z)\Delta
_{M}(\bar{z})\right] ^{\gamma }\prod_{j=1}^{M}w_{\text{DF}}(z_{j},\bar{z}%
_{j})\mathrm{d}z_{j}\mathrm{d}\bar{z}_{j} \\
& =\left( \mathcal{Z}_{M}^{\text{JE},\gamma }\right) ^{2}~\prod_{j=1}^{M}%
\frac{\sin \pi \left( \alpha +1+\frac{\gamma }{2}(M-j)\right) \sin \pi
\left( \beta +1+\frac{\gamma }{2}(M-j)\right) }{\sin \pi \left( \alpha
+\beta +\frac{\gamma }{2}(2M-j-1)\right) }\frac{\sin \pi \left( \frac{\gamma 
}{2}j\right) }{\sin \frac{\pi \gamma }{2}},
\end{align*}%
$\Re \gamma \geq 0$. In the second line, we have introduced the Jacobi
beta-ensemble with beta-parameter $\gamma $, in which the Vandermonde factor
appears as $\lvert \Delta _{M}(z)\rvert^{\gamma }$.\footnote{In particular, for $\gamma =1$ in the Dotsenko--Fateev integral we need the
JOE.} The corresponding partition function is the celebrated Selberg
integral \cite{Selberg,ForWar}: 
\begin{equation*}
\mathcal{Z}_{M}^{\text{JE},\gamma }=\frac{1}{M!}\prod_{j=1}^{M}\frac{\Gamma
\left( \alpha +1+\frac{\gamma }{2}(M-j)\right) \Gamma \left( \beta +1+\frac{%
\gamma }{2}(M-j)\right) }{\Gamma \left( \alpha +\beta +\frac{\gamma }{2}%
(2M-j-1)\right) }\frac{\Gamma \left( 1+\frac{\gamma }{2}j\right) }{\Gamma
\left( \frac{\gamma }{2}\right) }.
\end{equation*}

As a corollary of \cite[Thm.14]{SI} we get 
\begin{equation*}  \label{eq:SchurDF}
\left\langle s_{\lambda^{\prime}} \bar{s}_{\mu^{\prime}} \right\rangle_{%
\text{DF}} = \left\langle s_{\lambda^{\prime}} \right\rangle_{\text{JE},
\gamma} \left\langle s_{\mu^{\prime}} \right\rangle_{\text{JE}, \gamma} .
\end{equation*}
Thus, the following factorization theorem holds.

\begin{proposition}
The kernel on $\mathbb{C}$ associated to the Dotsenko--Fateev ensemble is
completely factorized into holomorphic and anti-holomorphic sector: 
\begin{equation*}
\widehat{K}_N ^{(n)} \left( x_1 , \dots , x_n ; \bar{y}_1 , \dots, \bar{y}_n
\right) \vert_{\mathrm{DF}} = \mathsf{K}_{N} ^{(n)} (x) \cdot \mathsf{%
K}_{N} ^{(n)} (\bar{y}) ,
\end{equation*}
where 
\begin{equation*}
\mathsf{K}_{N} ^{(n)} (x) \equiv \sum_{\lambda \in \mathbb{Y}_{n, N-n}}
s_{\lambda} \left(x^{\vee} \right) \left\langle s_{\lambda^{\prime}}
\right\rangle_{\mathrm{JE}, \gamma } .
\end{equation*}
\end{proposition}

The DF kernels we are considering are tightly related to conformal blocks in
the $SL(2, \mathbb{R})$ WZW conformal field theory \cite%
{Iguri:2007af,Iguri:2009cf}. In particular, for the 2-point kernel, it
follows from \cite{Iguri:2007af} (see also \cite{SI}) that 
\begin{align*}
\mathsf{K}_{N} ^{(1)} (z) & = \sum_{k=0} ^{N-1} h_k (-z^{-1}) \left\langle
e_k \right\rangle_{\mathrm{JE}, \gamma }  \notag \\
&= \sum_{k=0} ^{N-1} (-z)^{-k} {\binom{N-1 }{k }} \prod_{j=1} ^{k} \frac{
\alpha + 1 + \frac{\gamma}{2} (N-j-1) }{ \alpha + \beta +2 + \frac{\gamma}{2}
(2N-j-3) }
\end{align*}
where $e_k$ are elementary symmetric polynomials. The 2-point DF kernel is
thus reinterpreted as a generating function of certain correlation functions
in the WZW model, and the factorization follows from the ``chiral'' factorization property studied in \cite{Iguri:2007af,Iguri:2009cf}.

\section{Outlook}
\label{sec:outlook}

While we have studied a considerable number of cases explicitly, including
novel analytical evaluations, it is manifest that a more exhaustive study of
the expansions obtained can be carried out. It is worth mentioning that the
kernels discussed here appear in different contexts, sometimes without any
reference to random matrix theory. For example, in the very recent \cite{Klus}, a family of antisymmetric kernels is constructed. Due to the reproducing property, these kernels are essentially of the type \eqref{eq:defmultikernel} if the \textit{seed} kernel is a Christoffel--Darboux kernel.
Therefore, the expansions obtained here could conceivably be applied in this other context. It is then a natural question to ask if such an expansion could provide any type of, say computational, advantage.

If the seed kernel is a simpler one, such as a Gaussian kernel, an instance studied in detail in \cite{Klus}, we still have a multi-faceted connection with
random matrix theory, as explained in \cite[Sec.4]{ST}. The corresponding kernel is then the Karlin--McGregor kernel for
non-intersecting diffusion processes. In general, the
antisymmetrization procedure parallels the Karlin--McGregor construction of non-intersecting diffusion processes.\footnote{Likewise, if the base space is discretized, i.e. $\mathbb{R}^d$ is replaced with a lattice $a \mathbb{Z}^d$, the antisymmetric Gaussian kernel in \cite{Klus} can be written as the average of two Schur polynomials in a Stieltjes--Wigert or Rogers--Szeg\H{o} ensemble. The evaluation of such average is known and is given by a certain topological knot invariant in $\mathbb{S}^3$ \cite{DGMT}.}

This highlights yet another universality and interdisciplinarity aspect, inherent to such kernels. It would be interesting if some of the analytical results here have applications along these lines. A seemingly simpler open problem would be to eventually interpret the various Schur averages as transition probabilities of different systems of non-intersecting walkers.\par

To conclude, it is worth to mention that expanding the multi-point kernel in the Schur basis is certainly not the only option. One may expand in any basis of symmetric function as, for instance, in power sums in $t$. Such expansions follow from Theorem \ref{thm:main} with a change of basis. Let $\left\{ u_{\mu} \right\}_{\mu}$ be any basis in the ring of symmetric functions in $2n$ variables and let $C_{\lambda \mu}$ denote the entries of the transition matrix to the Schur basis, $s_{\lambda} = \sum_{\mu} C_{\lambda \mu} u_{\mu}$. Then, Theorem \ref{thm:main} implies that 
    \begin{equation*}
        \widehat{K}_N ^{(n)} (x_1, \dots, x_n; y_1, \dots, y_n) = \sum_{\mu \in \mathbb{Y}} u_{\mu} (t) \left[ \sum_{ \nu \in \mathbb{Y}_{2n,N-n} } \left\langle s_{\nu } \right\rangle_{w}   C_{\nu^{\prime} \mu} \right] .
    \end{equation*}
    The sum in square brackets is assembled into the average of a symmetric function (labelled by $\mu$), distinct from $u_{\mu}$ unless $u_{\mu} = s_{\mu}$. In general, these averages are not known in closed form.\par
    Alternatively, one could also expand the determinant $\det \prod_{i=1}^{2n} (1+ t_i Z)$ in the matrix average into products of traces of powers of the $M \times M$ matrix $Z$. Then, the coefficients in the expansion of the multi-point kernels would be given by (intricate combinations of) correlators of moments of random matrices, that appear in a broad variety of areas.\par

\vspace{0.3cm}
\subsection*{Acknowledgements}
We thank Prof. Peter Forrester for correspondence and insightful comments. The work of LS is supported by the Funda\c{c}\~{a}o para a Ci\^{e}ncia e a Tecnologia (FCT) through the doctoral grant SFRH/BD/129405/2017. The work is also supported by FCT Project PTDC/MAT-PUR/30234/2017.
\vspace{0.4cm}
\appendix

\section{Proofs of Schur polynomial averages}
\label{app:ComputeSchur}

\subsection{Schur polynomial average in the Laguerre ensemble}
\label{app:SchurClassic} 

In this appendix we give a proof of Proposition \ref%
{prop:LUESchur}, following \cite[Sec.4.2.2]{DavidGG}. Along the way, we keep
the discussion general to show how the method is suitable to the averages of
a Schur polynomial in classical ensembles.

Starting with the definitions 
\begin{equation*}
\langle s_{\mu} \rangle = \frac{ \int \Delta (z)^2 s_{\mu} (z)
~\prod_{j=1}^{M} w(z_j) \mathrm{d} z_j }{ \int \Delta (z)^2 ~\prod_{j=1}^{M}
w(z_j) \mathrm{d} z_j }
\end{equation*}
and 
\begin{equation*}
\Delta (z)^2 s_{\mu} (z) = \det_{1 \le j,k \le M } [x_k ^{M-j}] ~ \det_{1
\le j,k \le M } [x_k ^{\mu_j+M-j}] ,
\end{equation*}
we can apply Andreief's identity \cite{Andreief,Forrester:Meet} to write 
\begin{equation*}
\langle s_{\mu} \rangle = \frac{ \det_{1 \le j,k \le M } \left[ \mathfrak{m}%
_{\mu_j+M-j-1+k} \right]}{ \det_{1 \le j,k \le M } \left[ \mathfrak{m}%
_{j+k-2} \right] } ,
\end{equation*}
where we have adopted the shorthand notation 
\begin{equation*}
\mathfrak{m}_{p} \equiv \int z^{p} w(z) \mathrm{d} z
\end{equation*}
for the moments of the measure $w(z) \mathrm{d} z $. They can be evaluated
exactly for the classical ensembles of Section \ref{sec:ClassicCDK}: 
\begin{align*}
\mathfrak{m}_{p} ^{\mathrm{GUE}} & = \int_{- \infty} ^{+ \infty} \frac{ 
\mathrm{d} z }{ \sqrt{2\pi} } z^{p} e^{- z^2 /2 } = \frac{ 2^{p/2} }{\sqrt{%
\pi}} \left( 1 + (-1)^{p} \right) \Gamma \left( \frac{p+1}{2} \right) , \\
\mathfrak{m}_{p} ^{\mathrm{LUE}} & = \int_{0} ^{+ \infty} \mathrm{d} z z^{p+
\alpha} e^{- z} = \Gamma \left( 1 + \alpha +p \right) , \\
\mathfrak{m}_{p} ^{\mathrm{JUE}} & = \int_{0} ^{1} \mathrm{d} z z^{p+ \alpha}
(1-z)^{\beta} = \frac{\beta }{ (1+\alpha+p)(2+\alpha+p) } , \\
\mathfrak{m}_p ^{\widetilde{\mathrm{JUE}}} &= \int_0 ^{+\infty} \dd z \frac{z^{p + \alpha}}{(1+z)^{M+\beta}}  = \frac{\Gamma (p+\alpha +1) \Gamma (M-p-\alpha +\beta -1)}{\Gamma (M+\beta )} , \\
\mathfrak{m}_p ^{\widetilde{\mathrm{LUE}}} &= \int_0 ^{+\infty} \dd z e^{-1/z} z^{p-\tilde{\alpha}} = \Gamma (\tilde{\alpha} -p -1) .
\end{align*}
The property $\Gamma (z+1) =z \Gamma (z) $ can be used to simplify the determinants. We have the relations 
\begin{align*}
\mathfrak{m}_{\mu_j+M-j+k+1} ^{\mathrm{GUE}} &= \left( \mu_j+M-j+k \right) 
\mathfrak{m}_{\mu_j+M-j+k-1} ^{\mathrm{GUE}}  \hspace{50pt} \mathfrak{m}_{j+k} ^{\mathrm{GUE}} = \left( j+k-1 \right) \mathfrak{m}_{j+k-2}^{\mathrm{GUE}} , \\
\mathfrak{m}_{\mu_j+M-j+k} ^{\mathrm{LUE}} & = \left( \alpha + \mu_j+M-j+k \right) \mathfrak{m}_{\mu_j+M-j+k-1} ^{\mathrm{LUE}} \hspace{24pt} \mathfrak{m}_{j+k-1} ^{\mathrm{LUE}} = \left( \alpha + j+k \right) \mathfrak{m}_{j+k-2} ^{\mathrm{LUE}} , \\
\mathfrak{m}_{\mu_j+M-j+k} ^{\mathrm{JUE}} & = \frac{ \alpha + \mu_j+M-j+k }{ \alpha + \mu_j+M-j+k+2} \mathfrak{m}_{\mu_j+M-j+k-1} ^{\mathrm{JUE}} \hspace{14pt} 
\mathfrak{m}_{j+k-1} ^{\mathrm{JUE}} = \frac{ \alpha + j+k -1 }{ \alpha + j+k+1 } \mathfrak{m}_{j+k-2} ^{\mathrm{JUE}} , \\
\mathfrak{m} ^{\widetilde{\mathrm{JUE}}} _{\mu_j+M-j+k} & = \frac{\alpha + \mu_j+M-j+k}{  \beta - \alpha -1 -\mu_j+j-k} \mathfrak{m} ^{\widetilde{\mathrm{JUE}}} _{\mu_j+M-j+k-1} \hspace{20pt} \mathfrak{m} ^{\widetilde{\mathrm{JUE}}} _{j+k-1}  = \frac{\alpha +j+k-1}{ M+ \beta - \alpha -j-k} \mathfrak{m} ^{\widetilde{\mathrm{JUE}}} _{j+k-2} , \\
\mathfrak{m} ^{\widetilde{\mathrm{LUE}}} _{\mu_j+M-j+k} & = \left( \tilde{\alpha} - \mu_j -M +j -k -2 \right)  \mathfrak{m} ^{\widetilde{\mathrm{LUE}}} _{\mu_j+M-j+k-1} \hspace{10pt} \mathfrak{m} ^{\widetilde{\mathrm{LUE}}} _{j+k-1}  = \left( \tilde{\alpha} -j -k -1 \right)  \mathfrak{m} ^{\widetilde{\mathrm{LUE}}} _{j+k-2} .
\end{align*}

We henceforth focus on the LUE. We bring out the factor in the $(j,1)$-entry, $\forall j=1, \dots, M $ both in the numerator and the denominator.
This leaves the numerator as the determinant of the matrix whose entry $%
(j,k) $ is $\prod_{s=1}^{k-1} \left( \alpha + \mu_j+M-j+s \right) $, with $\prod_{s=1} ^{0} \equiv 1$ understood. A similar expression is obtained in the denominator.

Then, we successively subtract lower columns to the $k^{\text{th}}$ column,
so to reorganize the expression as 
\begin{equation*}
\langle s_{\mu} \rangle_{\text{LUE}} = \prod_{j=1} ^{M} \frac{ \Gamma \left(
\alpha + \mu_j + M -j +1 \right) }{ \Gamma \left( \alpha +j \right) } \cdot 
\frac{ \det_{1 \le j, k \le M } \left[ \left( \mu_j+M-j+1 \right)^{k-1} %
\right] }{ \det_{1 \le j, k \le M } \left[ \left( j \right)^{k-1} \right] } .
\end{equation*}
The determinant in the denominator is simply a Vandermonde on the integers,
that gives $G(M+1)$, while the numerator is a Vandermonde on the lattice $%
\mu_j+M-j+1$, which gives $G(M+1) s_{\mu} (1, \dots, 1)$. The ratio of determinants leaves behind $s_{\mu} (1^M)$. This proves Proposition \ref{prop:LUESchur} \cite{DavidGG}.

\begin{remark}
Identifying the integral with the insertion of a Schur polynomial with a
determinant of a minor of a Hankel matrix makes manifest that, whenever a
weight function $w(z)$ is even and with even support, the average $%
\left\langle s_{\mu} \right\rangle_{w}$ will be subject to parity
constraints on the partition $\mu$, precisely as in the GUE. This stems from
the vanishing of the odd moments of the measure $w(z) \mathrm{d} z $, as for instance in the Chebyshev weights of first and second kind.
\end{remark}

\subsection{Schur polynomial average in $q$-ensembles: Stieltjes--Wigert}

\label{app:SchurqSW}

This appendix contains a proof of formula \eqref{eq:SWSchur} for $\langle
s_{\mu} \rangle $ in the Stieltjes--Wigert~ensemble which is along the lines
of Appendix \ref{app:SchurClassic}. See \cite{Dolivet:2006ii} for a
different proof.

Reasoning as in Appendix \ref{app:SchurClassic}, we write 
\begin{equation*}
\langle s_{\mu }\rangle _{\text{SW}}=\frac{\det_{1\leq j,k\leq M}\left[ 
\mathfrak{m}_{\mu _{j}+M-j+k-1}^{\text{SW}}\right] }{\det_{1\leq j,k\leq M}%
\left[ \mathfrak{m}_{j+k-2}^{\text{SW}}\right] }.
\end{equation*}%
In the Stieltjes--Wigert~case the moments are 
\begin{equation*}
\mathfrak{m}_{p}^{\text{SW}}=\int_{0}^{\infty }z^{p}e^{-\frac{1}{2g}(\ln
z)^{2}}\frac{\mathrm{d}z}{\sqrt{2\pi g}}=q^{-\frac{(p+1)^{2}}{2}}
\end{equation*}%
where we recall that $q=e^{-g}$. We then extract the factor $q^{-\frac{1}{2}%
\left( \mu _{j}+M-j+1\right) ^{2}}$ from the $j^{\text{th}}$ row in the
determinant in the numerator, and the factor $q^{-j^{2}/2}$ from the $j^{%
\text{th}}$ row in the determinant in the denominator. We are left with 
\begin{equation*}
\langle s_{\mu }\rangle _{\text{SW}}=\prod_{j=1}^{M}\frac{q^{-\frac{1}{2}%
\left( \mu _{j}+M-j+1\right) ^{2}}}{q^{-\frac{j^{2}}{2}}}\cdot \frac{%
\det_{1\leq j,k\leq M}\left[ q^{-\left( \mu _{j}+M-j+1\right) (k-1)}q^{-%
\frac{(k-1)^{2}}{2}}\right] }{\det_{1\leq j,k\leq M}\left[ q^{-j(k-1)}q^{-%
\frac{(k-1)^{2}}{2}}\right] }.
\end{equation*}%
After bringing out the common factor from the $k^{\text{th}}$ row in both
the determinants, we recognize the ratio of a Vandermonde determinant
evaluated at the exponential lattice $q^{-\left( \mu _{j}+M-j+1\right) }$ in
the numerator, and a Vandermonde determinant on the exponential lattice $%
q^{-j}$ in the denominator. Simplifying also the prefactor, we get 
\begin{equation*}
\langle s_{\mu }\rangle _{\text{SW}}=q^{-\sum_{j=1}^{M}\left( \frac{\mu
_{j}^{2}}{2}+\mu _{j}(M-j+1)\right) }\prod_{1\leq j<k\leq M}\frac{q^{-(\mu
_{j}+M-j+1)}-q^{-(\mu _{k}+M-k+1)}}{q^{-j}-q^{-k}}.
\end{equation*}%
By direct computation, the last product can be simplified into 
\begin{equation}
\prod_{1\leq j<k\leq M}\frac{q^{-(\mu _{j}+M-j+1)}-q^{-(\mu _{k}+M-k+1)}}{%
q^{-j}-q^{-k}}=\prod_{j=1}^{M}q^{-\frac{(M-1)}{2}\mu _{j}}\prod_{1\leq
j<k\leq M}\frac{\left[ \mu _{j}-j-\mu _{k}+k\right] _{q}}{\left[ j-k\right]
_{q}}  \label{eq:ProductToQDim}
\end{equation}%
where in the last line we have used the symmetric $q$-number $\left[ \cdot %
\right] _{q}$, see \eqref{eq:defqnumber}. Recognizing the $q$-dimension \eqref{eq:defqdim}, Proposition \ref{prop:SWSchur} follows.

\subsection{Schur polynomial average in $q$-ensembles: $q$-Laguerre}

\label{app:SchurqLUE} Consider now the $q$-Laguerre weight. In this case the
moments are \cite{Moak} 
\begin{equation*}
\mathfrak{m}_{p}^{q\text{LUE}}=\int_{0}^{\infty }\mathrm{d}z\frac{%
z^{p+\alpha }}{(-(1-q)z;q)_{\infty }}=\frac{\Gamma (-p-\alpha )\Gamma
(p+\alpha +1)}{\Gamma _{q}(-p-\alpha )}.
\end{equation*}%
Under a shift $p\mapsto p+1$ they transform simply as 
\begin{equation*}
\mathfrak{m}_{p+1}^{q\text{LUE}}=-\lfloor -(p+\alpha +1)\rfloor _{q}%
\mathfrak{m}_{p}^{q\text{LUE}},
\end{equation*}%
where $\lfloor z\rfloor _{q}$ is the asymmetric $q$-number 
\begin{equation*}
\lfloor z\rfloor _{q}\equiv \frac{1-q^{z}}{1-q}.
\end{equation*}%
This closely resembles the behaviour of the LUE moments in Appendix \ref%
{app:SchurClassic}, with ordinary numbers $(p+\alpha +1)$ replaced by their $%
q$-analogue $-\lfloor -(p+\alpha +1)\rfloor _{q}$.

We can therefore proceed as in Appendix \ref{app:SchurClassic}, writing $%
\langle s_{\mu} \rangle_{q\text{LUE}}$ as a ratio of determinants and
extract the common factor in the $j^{\text{th}}$ row, $\forall j=1, \dots, M$%
. We arrive at 
\begin{align*}
\langle s_{\mu} \rangle_{q\text{LUE}} &= \prod_{j=1} ^{M} \frac{ \Gamma
\left( \alpha +1 + \mu_j + M -j \right) \Gamma (- \alpha - \mu_j - M +j)}{
\Gamma (\alpha +j ) \Gamma (- \alpha -j +1) } \frac{ \Gamma_q (- \alpha -j
+1) }{ \Gamma_q (- \alpha - \mu_j - M +j) } \\
& \times \frac{ \det_{1 \le j, k \le M} \left[ \prod_{s=1}^{k-1} (- \lfloor
- ( \alpha + \mu_j + M -j +s ) \rfloor_{q} ) \right] }{ \det_{1 \le j, k \le
M} \left[ \prod_{s=1}^{k-1} (- \lfloor - ( \alpha + j-1 +s ) \rfloor_{q} ) %
\right] }
\end{align*}
(with $\prod_{s=1} ^{0} (\cdots ) \equiv 1$ understood). Using the simple
recursion $- \lfloor - z \rfloor_{q} = q^{-1} (1- \lfloor - (z-1)
\rfloor_{q})$ and taking linear combinations of the columns, we can
rearrange the determinant in the numerator such that the $(j,k)$-entry is\footnote{The sign in front of each entry requires care. There is a product $\prod_{s=1} ^{k-1}$ and each entry of the product includes a factor $(1- \lfloor - (\cdots) \rfloor_q)$ if $s$ is odd and $\lfloor - (\cdots) \rfloor_q$ if $s$ is even, $2 \le s \le k-1$, while $s=1$ always contributes $- \lfloor - (\cdots) \rfloor_q$.} 
\begin{equation*}
- (-q)^{\frac{k(k-1)}{2}} \lfloor - (\alpha + \mu_j + M -j+1) \rfloor_{q}
^{k-1} ,
\end{equation*}
and similarly in the denominator with the usual replacement $\mu_j + M -j
\mapsto j-1$. Bringing out the common factor $- (-q)^{\frac{k(k-1)}{2}} $
from the $k^{\text{th}}$ column, $k = 2, \dots, M$, both in the numerator
and denominator, we are left with a ratio of Vandermonde determinants, over
the exponential lattices $\lfloor - (\alpha + \mu_j + M -j+1) \rfloor_{q} $
and $\lfloor - (\alpha +j) \rfloor_{q}$ respectively.

Putting all together we have 
\begin{align*}
\langle s_{\mu} \rangle_{q\text{LUE}} &= \prod_{j=1} ^{M} \frac{ \Gamma
\left( \alpha +1 + \mu_j + M -j \right) \Gamma (- \alpha - \mu_j - M +j)}{
\Gamma (\alpha +j ) \Gamma (- \alpha -j +1) } \frac{ \Gamma_q (- \alpha -j
+1) }{ \Gamma_q (- \alpha - \mu_j - M +j) } \\
& \times \prod_{1 \le j < k \le M } \frac{ \lfloor - ( \alpha + \mu_j -j
+M+1) \rfloor_{q} - \lfloor - ( \alpha + \mu_k -k +M+1) \rfloor_{q}}{
\lfloor \alpha + j \rfloor_{q} - \lfloor \alpha + k \rfloor_{q} } .
\end{align*}
Writing 
\begin{multline*}
\prod_{1 \le j < k \le M } \frac{ \lfloor - ( \alpha + \mu_j -j +M+1)
\rfloor_{q} - \lfloor - ( \alpha + \mu_k -k +M+1) \rfloor_{q}}{ \lfloor
\alpha + j \rfloor_{q} - \lfloor \alpha + k \rfloor_{q} } \\ = \prod_{1 \le j <
k \le M } \frac{ q^{-(\mu_j + M-j+1)} - q^{-(\mu_k + M-k+1)} }{ q^{-j} -
q^{-k} }
\end{multline*}
and using again \eqref{eq:ProductToQDim} to identify the $q$-dimension, we
arrive at Proposition \ref{prop:qLagSchur}.

\section{Duduchava--Roch formula for Toeplitz inverses and kernels}

\label{app:Duduchava}

In this appendix we do not focus on matrix models defined on (subsets of) $%
\mathbb{R}$, associated to Hankel determinants, and discuss instead a result
concerning Toeplitz matrices and their inverses, showing a relationship with
random matrix kernels. In this way, we give a different expansion of such
kernels, in a monomial basis, instead of in a Schur basis, as in the main
text.

Let $\mathcal{T} (f)$ be the Toepliz matrix with symbol $f: \mathbb{S}^1 \to 
\mathbb{C}$. Besides, let $\mathcal{M}_z$ be the diagonal matrix 
\begin{equation*}
\left( \mathcal{M}_z \right)_{jk} = \frac{ \Gamma (z+k)}{ \Gamma (k) \Gamma
(z+1) } ~ \delta_{jk} , \qquad z \in \mathbb{C} .
\end{equation*}
A pure Fisher--Hartwig~singularity is the weight function on the circle
defined as 
\begin{subequations}
\label{eq:defFSsing}
\begin{align}
w_{\mathrm{FH}} (z) & = \omega_{\gamma } (z) \widetilde{\omega}%
_{\delta} (z), \\
\omega_{\gamma} (z) & = (1- z/z_0)^{\gamma} , \qquad \widetilde{\omega}%
_{\delta} (z) = (1- z_0/z)^{\delta} ,
\end{align}
\end{subequations}
with $\Re (\gamma + \delta) >-1$ and $z \in \mathbb{S}^1$, and for a fixed
reference point $z_0 \in \mathbb{S}^1$ that is usually set to $1$.

\begin{theorem}[Duduchava--Roch \protect\cite{Duduchava,Roch,Bottcher:DR}]
If $\gamma, \delta, \gamma+\delta \notin \mathbb{Z}_{<0}$, it holds that 
\begin{equation}  \label{eq:DRthm}
\mathcal{T} \left( \omega_{\gamma} \right) \mathcal{M}_{\gamma + \delta} 
\mathcal{T} \left( \widetilde{\omega}_{\delta} \right) = \frac{ \Gamma
\left( 1+ \gamma \right) \Gamma \left( 1+ \delta \right) }{ \Gamma \left( 1+
\gamma + \delta \right) } \mathcal{M}_{\delta } \mathcal{T} \left(
\omega_{\gamma}\widetilde{\omega}_{\delta} \right) \mathcal{M}_{\gamma } ,
\end{equation}
where $\omega_{\gamma}, \widetilde{\omega}_{\delta}$ are defined in %
\eqref{eq:defFSsing}.
\end{theorem}

The identity \eqref{eq:DRthm} is known as the Duduchava--Roch formula. See 
\cite{Bottcher:DR} for extensive discussion and proofs.

Let us assume $\gamma, \delta \in \mathbb{Z}_{\ge 0}$, and also set $z_0=1$
to lighten the formulae. Our interest is in the inverse of the $M \times M$
Toeplitz matrix $\mathcal{T}_M \left( \omega_{\gamma}\widetilde{\omega}%
_{\delta} \right)$. It can be computed using \eqref{eq:DRthm}, obtaining 
\cite{Bottcher:DR,GGT1} 
\begin{equation}  \label{eq:Tinversejk}
\left[ \mathcal{T}_M ^{-1} \left( \omega_{\gamma}\widetilde{\omega}_{\delta}
\right) \right]_{jk} = (-1)^{j+k} \frac{ \Gamma (\gamma + j) \Gamma (\delta
+ k) }{ \Gamma ( j) \Gamma (k) } \sum_{r= \max \left\{ j,k \right\} } ^{M-1} 
\frac{ \Gamma (r) }{ \gamma (\gamma + \delta + r) } {\binom{ \gamma +r -k -1 
}{r-k }} {\binom{ \delta +r -j -1 }{r-j }} .
\end{equation}
It was pointed out in \cite[Eq.(2.34)]{GGT1} that 
\begin{equation}  \label{eq:TinverseElementary}
(-1)^{j+k} \left[ \mathcal{T}_{M+1} ^{-1} \left( \omega_{\gamma}\widetilde{%
\omega}_{\delta} \right) \right]_{jk} = \frac{ \mathcal{Z}_{M} ^{\gamma,
\delta}}{ \mathcal{Z}_{M+1} ^{\gamma, \delta} } \left\langle e_{j-1} \bar{e}%
_{k-1} \right\rangle_{\mathrm{FH}} \vert_{M},
\end{equation}
with $\mathcal{Z}_{M} ^{\gamma, \delta}$ the partition function of the
Fisher--Hartwig~ensemble of $M$ variables. Therefore, \eqref{eq:Tinversejk}
implies the evaluation of the Selberg--Morris integral 
\begin{equation*}
\mathcal{Z}_{M} ^{\gamma, \delta} \left\langle e_{j-1} \bar{e}_{k-1}
\right\rangle_{\mathrm{FH}} \vert_{M} = \frac{1}{M!} \oint_{\left( 
\mathbb{S}^1 \right)^{M}} e_{j-1} (z) e_{k-1} (\bar{z}) \prod_{j=1} ^{M}
\omega_{\gamma} (z_j) \widetilde{\omega}_{\delta} (z_j) \frac{ \mathrm{d}
z_j }{2 \pi \mathrm{i} z_j}
\end{equation*}
where $e_{j-1} (z) \equiv e_{j-1} (z_1, \dots, z_{M})$ and likewise for $%
e_{k-1} (\bar{z})$.

At this point, we consider the 2-point kernel $K_N (x, \bar{y}) \vert_{\text{%
\textrm{FH}}}$ associated to the Fisher--Hartwig~weight function %
\eqref{eq:defFSsing}. With the conventions of Section \ref{sec:CKernels} and
using Theorem \ref{thm:double} we get 
\begin{align}
K_N (x, \bar{y}) \vert_{\mathrm{FH}} & = \sum_{j=0} ^{N-1} \sum_{k=0}
^{N-1} x^{N-j-1} \bar{y}^{N-k-1} (-1)^{j+k} \left\langle e_{j} \bar{e}_k
\right\rangle_{\mathrm{FH}}  \notag \\
& = \sum_{j=0} ^{N-1} \sum_{k=0} ^{N-1} x^{N-j-1} \bar{y}^{N-k-1} \left[ 
\mathcal{T}_N ^{-1} \left( \omega_{\gamma}\widetilde{\omega}_{\delta}
\right) \right]_{jk} .  \label{eq:DRvsFHkernel}
\end{align}
In passing to the second line we have used \eqref{eq:TinverseElementary}
with $M=N-1$. Notice that the ratio of partition functions from %
\eqref{eq:TinverseElementary} has cancelled against the change in
normalization from $\widehat{K}_N$ to $K_N$.

Therefore, \eqref{eq:DRvsFHkernel} shows that the kernel associated to the
Fisher--Hartwig~weight is the generating function of the inverse of the
Toeplitz matrix with symbol $w_{\mathrm{FH}}$. This fact is known
since long ago and is a particular case of \cite{Collar}, see also \cite%
{Kailath,SimonCD} for further discussion. Here we have given a proof based
only on the Duduchava--Roch formula \eqref{eq:DRthm}.

In conclusion, formulas \eqref{eq:DRvsFHkernel}-\eqref{eq:Tinversejk} give
explicitly the expansion of the Fisher--Hartwig~kernel. In turn, $K_{N}(x,%
\bar{y})|_{\mathrm{FH}}$ is related via Gessel's identity \cite{Gessel} to the Meixner kernel (see \cite{Santilli:2020ueh} and references
therein). The latter is a limit of the hypergeometric kernel \cite{BorOl2},
that has many connections with various objects in random matrix theory.

\section{Alternative derivation of the Chebyshev heat kernel}
\label{app:HKCheby}

In this work we have restricted our attention to Christoffel--Darboux kernels, which have no temporal dependence and hence can be understood as an ``initial time'' version of heat kernels (with no scaling, this is
equivalent to the delta sequence point of view of the Christoffel--Darboux kernels). However,
we can employ the little-used map between Schur polynomials of two variables and Chebyshev polynomials of the second kind, to evaluate the heat kernel associated to such polynomials using Schur-related techniques. This emphasizes the correspondence. Consider then the heat kernel \cite[Sec.3]{Allez} 
\begin{subequations}
\label{multieq:hk}
\begin{align}
\mathcal{K}_{q}(\xi ,\eta )& =\sum_{j=0}^{\infty }q^{j}U_{j}\left( \frac{\xi 
}{2}\right) U_{j}\left( \frac{\eta }{2}\right) \label{eq:Chebyheatsum}  \\
& =\frac{1-q^{2}}{1-q\xi \eta +q^{2}\left( \xi ^{2}+\eta ^{2}-2\right)
-q^{3}\xi \eta +q^{4}},  \label{eq:ChebyHeatKernel}
\end{align}
\end{subequations}
for $-2\leq \xi ,\eta \leq 2$ and $0<q<1$. The ``time'' is $-2\log q$. The heat kernel \eqref{multieq:hk} is
genuinely different from the kernels we have considered in the main text, whence the change in notation.

As in \eqref{eq:Chebyl2}, the relation 
\begin{equation*}
U_{j}\left( \frac{\xi }{2}\right)
=(x_{1}x_{2})^{-k-j/2}s_{(k+j,k)}(x_{1},x_{2}),\qquad \xi \equiv \frac{%
x_{1}+x_{2}}{\sqrt{x_{1}x_{2}}}
\end{equation*}%
between Schur and Chebyshev polynomials of second kind can be plugged in \eqref{eq:Chebyheatsum} to show
that 
\begin{align*}
\mathcal{K}_{q}(\xi ,\eta )& =\sum_{j=0}^{\infty
}q^{j}s_{(k+j,k)}(x,x^{-1})s_{(k+j,k)}(y,y^{-1}) \\
& =\sum_{j=0}^{\infty }q^{j}h_{j}(x,x^{-1})h_{j}(y,y^{-1}).
\end{align*}%
We have used the freedom in doubling the variables to set $%
x_{1}=x=x_{2}^{-1} $ and $y_{1}=y=y_{2}^{-1}$. At this point, we use the
equality 
\begin{equation*}
\sum_{j=0}^{\infty }q^{j}h_{j}(x,x^{-1})h_{j}(y,y^{-1})=\oint \frac{\mathrm{d}z}{2\pi {\,\mathrm{i}\,}z}\frac{1}{\left( 1-qzx\right) \left( 1-qz/x\right)
\left( 1-qzy\right) \left( 1-qz/y\right) },
\end{equation*}
which is the simplest case of Gessel's identity \cite{Gessel}. Computing the integral by residues reproduces \eqref{eq:ChebyHeatKernel}.

\end{document}